\RequirePackage{amsmath}
\documentclass[runningheads]{llncs}
\input{imports}

\DeclareMathOperator{\mono}{\hookrightarrow}

\DeclareMathOperator{\Mono}{\mathcal{M}}

\DeclareMathOperator{\Epi}{\mathcal{E}^{\prime}}
\DeclareMathOperator{\Shift}{Shift}
\DeclareMathOperator{\Left}{Left}

\DeclareMathOperator{\Tree}{Tr}
\DeclareMathOperator{\Leafs}{Leaves}
\DeclareMathOperator{\src}{src}
\DeclareMathOperator{\tar}{tar}

\DeclareMathOperator{\dpess}{DPEss}
\DeclareMathOperator{\ipess}{IPEss}
\DeclareMathOperator{\dess}{DEss}

\DeclareMathOperator{\ess}{ess}
\DeclareMathOperator{\pshift}{PEShift}

\DeclareMathOperator{\cess}{CEss}

\DeclareMathOperator{\cessAC}{SCEss}
\DeclareMathOperator{\comp}{Comp}
\DeclareMathOperator{\essdbl}{ess_{dbl}}

\newcommand{\rle}[5]{#1 \overset{#2}{\longleftarrow\joinrel\rhook} #3 \overset{#4}{\lhook\joinrel\longrightarrow}#5 }
\newcommand{\rleShort}[3]{#1 \hookleftarrow #2 \hookrightarrow #3}
\newcommand{\rleComplete}{\rle{L}{l}{K}{r}{R}}

\newcommand{\true}{\textsf{true}}

\newcommand{\rleAC}[1]{(\rle{L_#1}{l_#1}{K_#1}{r_#1} {R_#1}, ac_#1)}
\newcommand{\rleACCompact}[1]{(\rleShort{L_#1}{K_#1}{R_#1}, ac_#1)}

\newcommand{\rlePlain}[1]{\rle{L_#1}{l_#1}{K_#1}{r_#1} {R_#1}}

\newcommand{\trans}[4]{#1 \Longrightarrow_{#2,#3}#4}

\newcommand{\transPair}[7]{#1\, {}_{#2,#3}\!\Longleftarrow #4 \Longrightarrow_{#5,#6} #7}

\newcommand{\EMfactorization}{$\Epi$-$\Mono$ factorization} 

\newcommand{\DPEss}[3]{\dpess^{#1}(#2, #3)}
\newcommand{\IPEss}[3]{\ipess^{#1}(#2, #3)}
\newcommand{\Ess}[3]{\ess_{#1}(#2, #3)}
\newcommand{\DEss}[2]{\dess(#1, #2)}

\newcommand{\ESS}[2]{\dess(#1,\allowbreak #2)}

\newcommand{\PEShift}[2]{\pshift_{#1}(#2)}

\newcommand{\CEss}[2]{\cess(#1, #2)}

\newcommand{\Comp}[2]{\comp(#1, #2)}

\newcommand{\rmono}{right hook-stealth}
\newcommand{\lmono}{left hook-stealth}
\newcommand{\edge}[1]{\scriptsize $#1$}

\newcommand{\Edge}[5]{(m-#1) edge[#2] node[#3]{\edge{#4}} (m-#5)}

\newcommand{\Cedge}[4]{edge[#1] node[#2]{\edge{#3}} (m-#4)}
\newcommand{\CEssAC}[2]{\cessAC(#1,#2)}

\newcommand{\Span}[5]{#1 \xleftarrow{#2} #3 \xrightarrow{#4} #5}
\newcommand{\MonoSpan}[5]{#1 \xhookleftarrow{#2} #3 \xhookrightarrow{#4} #5}
\newcommand{\CoSpan}[5]{#1 \xrightarrow{#2} #3 \xleftarrow{#4} #5}
\newcommand{\MonoCoSpan}[5]{#1 \xhookrightarrow{#2} #3 \xhookleftarrow{#4} #5}

\definecolor{orange}{RGB}{255,145,0}

\newcommand{\decapAttribute}{\emph{decapsulateAttribute}}
\newcommand{\pullAttribute}{\emph{pullUpEncapsulatedAttribute}}

\begin{document}

	 \title{		
		Conflict Essences for Transformation Rules with Nested Application Conditions -- Long Version
	}
		\titlerunning{Conflict Essences for Nested Application Conditions}

	%
	\author{Alexander Lauer\inst{1}
	\orcidID{0009-0001-9077-9817} \and
		 Jens Kosiol\inst{1}\orcidID{0000-0003-4733-2777}\and
		 Leen Lambers\inst{2}\orcidID{0000-0001-6937-5167}\and
		Gabriele Taentzer\inst{1}\orcidID{0000-0002-3975-5238} } 
	\authorrunning{Lauer et al.}
	%
	\institute{
		Philipps-Universität Marburg, Marburg, Germany \\
		\email{alexander.lauer@uni-marburg.de} \\
		\email{\{kosiolje,taentzer\}@mathematik.uni-marburg.de}
		\and
		Brandenburg University of Technology Cottbus-Senftenberg, Cottbus, Germany \\
		\email{leen.lambers@b-tu.de} 
	}
	\maketitle              
	\begin{abstract}
	
		\emph{Conflict and dependency analysis} is an important static analysis tool that provides an overview of the potential interactions of (graph) transformation rules.
		This analysis is based on \emph{critical pairs} and \emph{initial conflicts}, which represent conflicting transformations in a minimal context. However, the crucial information about a conflicting transformation pair is contained in much smaller structures, called \emph{disabling/conflict essences} in existing research. 
		Recently, we introduced \emph{disabling essences for rules with application conditions} which contain the information on how an application condition can be violated by another rule. 
		In this paper, we extend the notion of disabling essences to support not only application conditions in \emph{Alternating Quantifier Normal Form}, but also arbitrary nested conditions. We introduce \emph{(symbolic) conflict essences} that are constructed from disabling essences and which capture the interaction between two rules. 
		We show that a transformation pair is parallel dependent if and only if a symbolic conflict essence can be embedded into it and
		relate symbolic conflict essences to initial conflicts for transformation rules with application conditions.
		We present our results for adhesive HLR categories, which includes several types of graph-like structures.

	\end{abstract}

	\section{Introduction}
	\label{chapter:introduction}
	The algebraic approach to graph transformation~\cite{EEPT06, EEGH15} has a broad range of interesting applications, especially in software engineering~\cite{HT20}.
\emph{Transformation rules} specify how graphs can be modified by either deleting existing elements or creating new elements. 
In addition, rules may be equipped with a \emph{nested application condition} to further model the situations in which a rule is allowed to be applied. 
When multiple rules are to be applied to the same graph, they can interact in the sense that one rule might prevent the application of the other by either deleting elements that are used in the other transformation or by inserting or deleting elements so that the application condition is not satisfied anymore. 
Then one rule is said to cause a \emph{conflict} for the other.

Statically analysing rules for these interactions plays an important role in applications such as feature interaction detection, model versioning, test case generation, and graph parsing (see~\cite{LSTBH18} for an overview of such applications). 
Each of these analysis techniques provides an overview of all potential conflicts for a given rule set with different types of granularities. 
\emph{Critical pairs} describe conflictual transformation pairs in a minimal context~\cite{EHPP04}, whereas \emph{initial conflicts} (a subset of the set of critical pairs of a rule pair) describe conflictual transformation pairs with minimal overlaps. 
Since critical pairs contain unnecessary information for many applications~\cite{LSTBH18}, research has been conducted to extract only the essential information about potential conflicts~\cite{LBKST19,LSTBH18,ACR19,LKST19,LauerKT25}.  
All notions developed for \emph{plain} rules (i.e., those without application conditions) are closely related as shown at the top of~\cref{fig:overview}. 
\begin{figure}[t]
	\centering\includegraphics[width = \textwidth]{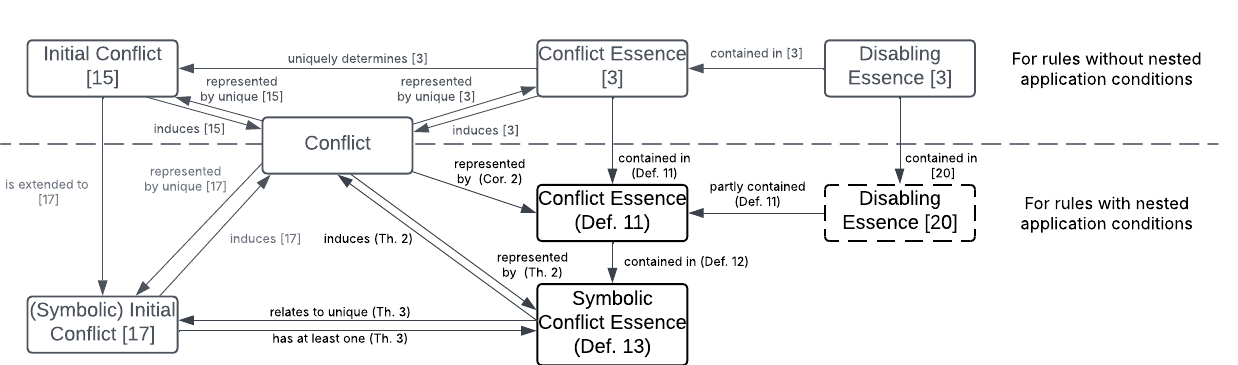}
	\vspace{-.5cm}
	\caption{Overview of concepts for conflicts between rules with and without nested application conditions (where square brackets contain literature references and round brackets refer to concepts and results of this paper).}
	\label{fig:overview}
	\vspace{-1.7em}
\end{figure}

\emph{Conflict reasons}~\cite{LBKST19} or \emph{conflict essences}~\cite{ACR19} show the core information of a conflict, namely, the elements whose deletion leads to a conflict.
Both approaches only consider plain rules. 
Since practical applications usually involve rules with conditions and these can be involved in conflicts, this is a significant gap in the research.

We have started to address this gap by introducing \emph{disabling essences} for rules with application conditions in \ac{ANF}~\cite{LauerKT25}. These essences describe how a rule can prevent the application of another rule by either deleting its used elements or by violating its application condition.
We showed that in each pair of conflicting transformations a disabling essence can be embedded. 
Conversely, an embedding of a disabling essence does not imply that the transformations are in conflict, which is a limitation of the theory. 
In this paper, we continue our work on disabling essences for rules with application conditions as follows: 
(1) We generalise our construction to support arbitrary nested application conditions. 
(2) We introduce \emph{conflict essences} for rules with application conditions, which describe situations in which both transformations cause a conflict for the other.  To achieve this, we present a construction for composing disabling essences and conflict essences.
(3) We define \emph{symbolic conflict essences}, which are \emph{conflict essences} with additional application conditions that ensure they are only embeddable in conflictual transformation pairs. 
(4) We relate symbolic conflict essences to \emph{initial conflicts} for rules with application conditions (see \cref{fig:overview}). 
Whereas contribution (1) is a rather straightforward extension of our work in~\cite{LauerKT25}, contributions (2)--(4) are results we did \emph{not} already obtain for conditions in \ac{ANF} there.

The paper is structured as follows:  \cref{chapter:relatedWork} presents related work on conflict analysis.  \Cref{chapter:example} presents the example used throughout the paper. 
\Cref{chapter:preliminaries} presents the formal preliminaries, \cref{chapter:essences} introduces our generalised construction for disabling essences, \cref{chapter:embedding} addresses conflict essences, and \cref{chapter:symbolic_conflict_essences} introduces symbolic conflict essences and relates them to initial conflicts.
Appendices contain all proofs (\cref{app:proofs}) and additional examples (\cref{app:detailed_examples}).

	\section{Related Work}
	\label{chapter:relatedWork}
	
We have thoroughly discussed related work on the topic of this paper in~\cite{LauerKT25}. 
Therefore, we here just focus on providing the immediate context of our work.

A \emph{critical pair} is a pair of conflicting transformations in a minimal context and the set of all critical pairs gives an overview of all conflicts of a graph transformation system (\emph{completeness}). 
Since critical pairs tend to be large objects, and even small transformation systems can induce large sets of critical pairs, two orthogonal research directions have been pursued: 
(1) reducing the number of critical pairs that must be computed while ensuring completeness and (2) identifying objects that contain the core information about conflicting transformations more compactly than critical pairs do. 
Both has been addressed first for \emph{plain} rules and subsequently also for rules with application conditions. 
Regarding (1), it has been found that (under certain technical conditions) so-called \emph{initial conflicts} are the smallest set of critical pairs that is still complete~\cite{LBOST18}; this also holds in the presence of application conditions (where initial conflicts then become \emph{symbolic})~\cite{LO21}. 
Regarding (2), i.e., computing the core information contained in a conflict, we~\cite{LBKST19} as well as Azzi et al.~\cite{ACR19} introduced \emph{(minimal) conflict reasons} and \emph{conflict essences}, respectively, as suitable such objects and also showed their unique correspondence to initial transformations. 
Importantly, these works only treat plain rules. 
That is, a notion of conflict essences and clarification of their relation to initial conflicts is still missing in the presence of application conditions. 
We started work in this direction in~\cite{LauerKT25}, were we lift the notion of disabling essences of Azzi et al. to the setting of rules with application conditions. 
However, we only addressed application conditions in \ac{ANF} and did not develop a composition of disabling essences to yield \emph{conflict essences} (considering the conflicts caused by each of the two transformations for the other simultaneously) or relate them to initial conflicts -- both of which we do in this paper for general application conditions.

	\section{Running Example}
	\label{chapter:example}
	To illustrate our work, we will consider two refactoring methods~\cite{fowler2018refactoring}.  
Typically, a sequence of refactorings is required to achieve a larger system design improvement.  Due to the potential of implicit conflicts and dependencies between refactorings, developers may have difficulties determining which refactorings to use and in what order.   The identification of potential conflicts and dependencies can help developers in specifying correct refactorings.

Assuming graphs that model the class design of software systems, we consider \cref{fig:examplerules} for two class model refactorings being specified as graph-based transformation rules.  These rules are depicted in an integrated form (as used, e.g., in the model transformation tool Henshin~\cite{ArendtBJKT10}), with annotations specifying which graph elements are deleted, preserved, or created.  The preserved and deleted elements form the \ac{LHS} of a rule,  and the preserved and created elements form the \ac{RHS}.  
Rule \emph{decapsulateAttribute} removes the getter and setter methods for a given attribute, thus inverting the well-known encapsulate attribute refactoring. Rule \emph{pullUpEncapsulatedAttribute} takes an attribute with its getter and setter methods and moves them to a superclass.

\begin{figure}[t]

	\centering 

	\includegraphics[scale=.35]{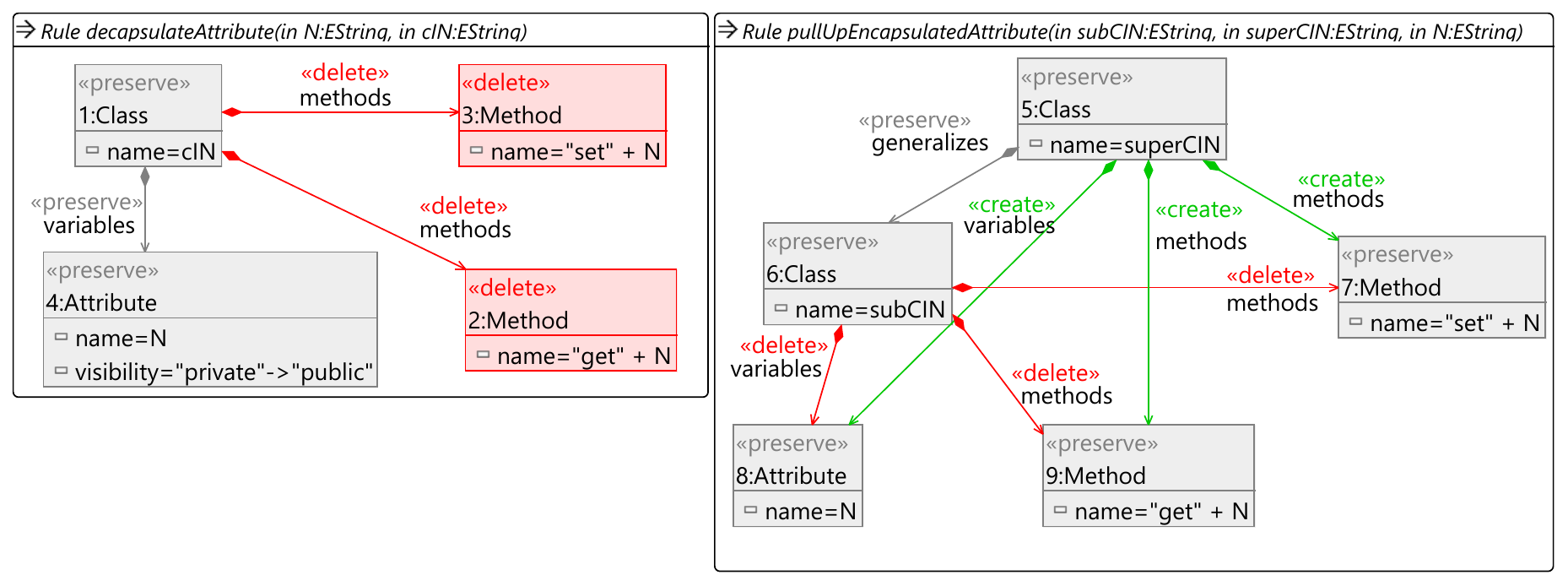}
	
	\vspace{-.7em}
	\caption{Henshin rules of our running example}
	\label{fig:examplerules}
\end{figure}

The rule \emph{pullUpEncapsulatedAttribute} has an additional nested application condition that is depicted in \cref{fig:exampleapplcond}.  This condition specifies that all subclasses of \textsf{5:Class} must also have a private attribute named \textsf{N} and corresponding getter and setter methods or a public attribute named \textsf{N}. 
In our presentation of the application condition, we omit morphisms and recurring parts of graphs.

\begin{figure}[t]
	
	\centering
	\includegraphics[scale=.55]{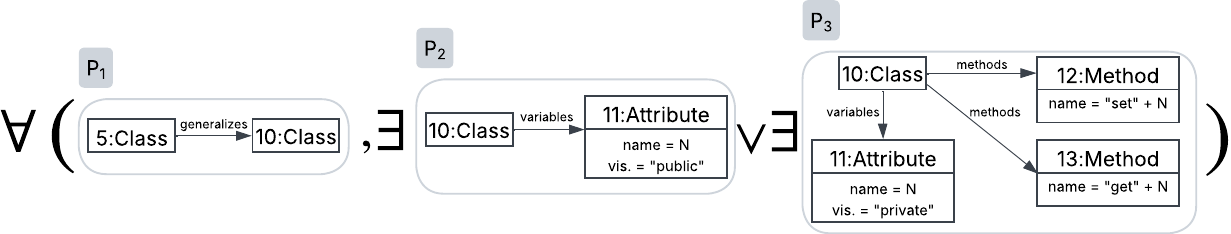}
	\vspace{-0.8em}
	\caption{Additional application condition for rule \emph{pullUpEncapsulatedAttribute}}
	\label{fig:exampleapplcond}
	\vspace{-1.5em}
\end{figure}

	\section{Preliminaries}
	\label{chapter:preliminaries}
	In this section, we present the preliminaries of our work. 
For reasons of space,  we omit standard background information and technical details that are required for the proofs, but which are not necessary for understanding the main ideas of the paper. 
In particular, we assume familiarity with basic category theory (including the concept of adhesiveness) and the double-pushout approach to graph transformation~\cite{LS05, EEPT06, EEGH15}.
A full introduction that covers all the necessary details can be found in~\cite{ACR19}. 
Here, we recall \emph{nested conditions} and concepts related to \emph{parallel independence}.

\emph{Throughout our work, we make the following assumptions: we work in an adhesive HLR category $(\mathcal{C}, \Mono)$ with \emph{all} pullbacks, initial pushouts over $\mathcal{M}$-mor\-phisms, an $\mathcal{M}$-initial object $\bot$, whose outgoing arrows are denoted by $!_A\colon \bot \hookrightarrow A$ for any object $A$, initial transformations, binary coproducts, and a unique \EMfactorization\ of cospans of morphisms.}

The semantics we assume for transformations in this paper is the double-push\-out approach
~\cite{EEGH15}. 
Thus, a \emph{transformation} $t\colon G \Longrightarrow_{\rho, m} H$ comes with a span $G \xhookleftarrow{g} D \xhookrightarrow{h} H$ of $\Mono$-morphisms that has been computed via two pushouts for some transformation rule $\rho = (\rleComplete, \mathit{ac})$ and a match $m$ of $\rho$ in $G$ (satisfying the application condition $\mathit{ac}$ of $\rho$); $G$ is the \emph{input (graph)} of $t$, the \emph{context (graph)} $D$ is computed as a pushout complement (by deleting the elements specified by $\rho$ from $G$), and $H$ is computed as a pushout (by creating the elements specified by $\rho$ on top of $D$). 
An \emph{ac-disregarding transformation} is one for which the match does not need to satisfy the rule's application condition.

\paragraph{Nested Conditions.} 
In the context of graph transformation, \emph{nested conditions and constraints} have been developed as a suitable (graphical) formalism for expressing properties of graphs and graph homomorphisms. 
They have been shown to be expressively equivalent to first-order logic on graphs~\cite{Rensink04, HabelP09}, but can be defined for arbitrary categories.

\begin{definition}[Nested condition]\label{def:conditions}
	A \emph{nested condition} over an object $P_0$ is defined recursively as follows: 
	$\true$ is a condition over $P_0$, 
	each Boolean combination of conditions over $P_0$ is a condition over $P_0$, 
	
	and $\exists\,(p_1 \colon P_0 \to P_1, d)$ is a condition over $P_0$ if $d$ is a condition over $P_1$.
	A morphism $q \colon P_0 \to G$ \emph{satisfies} a condition $c = \exists\,(p_1 \colon P_0 \to P_1,d)$, denoted by $q \models c$, if there is an $\Mono$-morphism $q_1 \colon P_1 \mono G \in \Mono$ so that $q = q_1 \circ p_1$ and $q_1 \models d$.
	Satisfaction for nested conditions composed by Boolean operators is defined as usual.
	We use the abbreviation
	$\forall\,(p_0 \colon P_0 \to P_1, d) \coloneqq \neg \exists\,(p_0 \colon P_0 \to P_1, \neg d)$.
	The object $P_1$ in a condition of the form $\exists\,(p_0 \colon P_0 \to P_1,d)$ is called \emph{existentially bound} and the object $P_1$ in a condition of the form $\forall\,(p_0 \colon P_0 \to P_1,d)$ is called \emph{universally bound}.
\end{definition}
In this paper, we make the following assumptions about all conditions: (i) All morphisms appearing in a condition are $\Mono$-morphisms, except potentially those starting at the object $P_0$ over which the condition is defined (i.e., the condition is in \emph{$\Mono$-normal form}); (ii) no isomorphisms appear in a condition; (iii) all negations are pushed inward as far as possible by eliminating double negations, replacing existential quantifiers with universal ones (and vice versa), and applying De Morgan's law (in this way, it is unambiguous whether an object is bound existentially or universally). 
\emph{These assumptions do not restrict the expressivity of the considered conditions;} the desired presentation can be achieved by applying well-established equivalence rules~\cite{HabelP09,Pennemann2009}. 

Conditions can be shifted along morphisms and left-translated along rules using operations called $\Shift$ and $\Left$,  respectively, in such a way that the semantics of the conditions are preserved~\cite{HabelP09,EEGH15}.  
We will use these operations in the following without recalling their definitions.

\paragraph{Parallel independence.}
Two \emph{parallel independent} transformations do not interfere with each other. This concept is formally captured by the Local Church--Rosser theorem~\cite{EEPT06}, which states that two such transformations can be applied in any order. 
We recall the definition of parallel independence and revisit the notion of a \emph{disabling essence}~\cite{ACR19} for analysing transformations that \emph{conflict} with each other, i.e., those that are not parallel independent. 

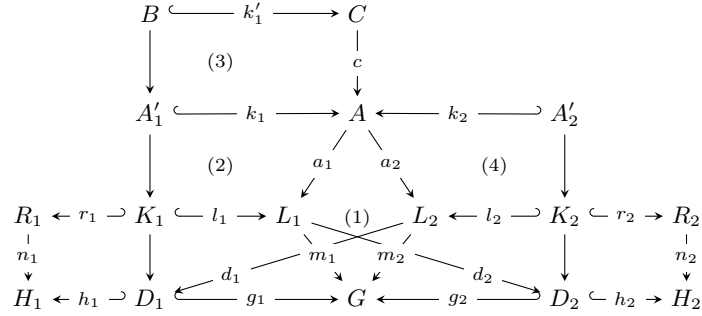
\begin{figure}[t]
	\centering
	\begin{tikzpicture}
			\matrix (m) [	matrix of math nodes,
		nodes in empty cells,
		row sep=.6em,
		column sep=1em,
		minimum width=1em]
		{
					& & B & & & C \\ 
					& & & \text{\scriptsize(3)}\\
			& & A_1' & & & A & & & A_2' \\ 
					& & & \text{\scriptsize(2)} & & & & \text{\scriptsize(4)}\\
			R_1 & & K_1 & & L_1 & \text{\scriptsize(1)} & L_2 & & K_2 & & R_2 \\ \\
			H_1 & & D_1 & &  & G_{} &  & & D_2 & & H_2 \\
			};
		\path[-stealth]
			\Edge{1-3}{\rmono}{fill=white}{k'_1}{1-6}
			\Cedge{}{}{}{3-3}
			\Edge{1-6}{}{fill=white}{c}{3-6}
			
			\Edge{3-3}{\rmono}{fill=white}{k_1}{3-6}
			\Cedge{}{}{}{5-3}
			
			\Edge{3-6}{}{fill=white}{a_1}{5-5}
			\Cedge{}{fill=white}{a_2}{5-7}
			
			\Edge{5-3}{\rmono}{fill=white}{l_1}{5-5}
			\Cedge{\lmono}{fill=white}{r_1}{5-1}
			\Cedge{}{}{}{7-3}
			
			\Edge{7-3}{\rmono}{fill=white}{g_1}{7-6}
			\Cedge{\lmono}{fill=white}{h_1}{7-1}
			
			\Edge{5-1}{}{fill=white}{n_1}{7-1}

			\Edge{3-9}{\lmono}{fill=white}{k_2}{3-6}
			\Cedge{}{}{}{5-9}
			
			\Edge{5-9}{\lmono}{fill=white}{l_2}{5-7}
			\Cedge{\rmono}{fill=white}{r_2}{5-11}
			\Cedge{}{}{}{7-9}
			
			\Edge{7-9}{\lmono}{fill=white}{g_2}{7-6}
			\Cedge{\rmono}{fill=white}{h_2}{7-11}
			
			\Edge{5-11}{}{fill=white}{n_2}{7-11}
			
			\Edge{5-7}{}{fill=white, near end}{d_1}{7-3}
			\Edge{5-5}{}{fill=white, near end}{d_2}{7-9}

			\Edge{5-5}{-,draw=white, line width=4pt}{}{}{7-6}
			\Edge{5-5}{}{fill=white}{m_1}{7-6}
			
			\Edge{5-7}{-,draw=white, line width=4pt}{}{}{7-6}
			\Edge{5-7}{}{fill=white}{m_2}{7-6}
		;
	\end{tikzpicture}
	\vspace{-1em}
	\caption{Parallel independence and construction of disabling essence}
	\label{fig:Construction_essences}
	\vspace{-1.5em}
\end{figure}

\begin{definition}[Parallel independence]\label{def:parallel_independence}
	Given two rules $\rho_j = \rleACCompact{j}$, with $j = 1,2$, two transformations $(t_1,t_2) = (\trans{G}{\rho_1}{m_1}{H_1}, \trans{G}{\rho_2}{m_2}{H_2})$ are \emph{parallel independent} if there are morphisms $d_1 \colon L_2 \to D_1$ and $d_2 \colon L_1 \to D_2$ to the contexts of $t_1$ and $t_2$ so that $m_2 = g_1 \circ d_1$, $m_1 = g_2 \circ d_2$, $h_1 \circ d_1 \models ac_2$, and $h_2 \circ d_2 \models ac_1$ (as shown in \cref{fig:Construction_essences}).
	If the morphism $d_1$ does not exist, or if $h_1 \circ d_1 \not \models ac_2$, we say that $t_1$ \emph{causes a conflict} for $t_2$.
	We say that 
	$t_1$ and $t_2$ are \emph{parallel independent ac-disregarding} if both morphisms $d_j$ exist.
\end{definition}

A transformation via a plain rule causes a conflict with another such transformation if it deletes items that the second transformation also requires. 
A \emph{disabling essence} captures exactly this problematic structure and its occurrence 
in conflicting transformations. 
The construction can be introduced for arbitrary cospans of the LHSs of rules, leading to the notion of a \emph{proto-essence}, whereas only matches for transformations lead to a \emph{disabling essence}. 

\begin{definition}[Proto-/disabling essence~\cite{ACR19}]\label{def:essential_conflicts}
	Given two plain rules $\rho_j = \rleShort{L_j}{K_j}{R_j}$ 
	and morphisms $m_j \colon L_j \to G$, with $j=1,2$, the \emph{proto-essence for $(m_1, m_2)$ under $l_1$}, denoted by $\ess_{l_1}(m_1, m_2)$, is defined as 
	$(a_1 \circ c, a_2 \circ c)$, where the squares $(1)$ and $(2)$ in \cref{fig:Construction_essences} are pullbacks and $(3)$ is the initial pushout over $k_1$. $C$ is called the \emph{essence object}. 
	A proto-essence $\ess_{l_1}(m_1, m_2)$ is a \emph{disabling essence} if $m_1$ and $m_2$ are matches for transformations $t_1$ and $t_2$ as in \cref{fig:Construction_essences}; in that case, it is denoted as $\essdbl(t_1,t_2)$ and called the \emph{disabling essence of $(t_1,t_2)$}. 
\end{definition}
For ease of notation,  we will sometimes drop the morphism $c$ in proto-essences if it is not needed.
Crucially, transformation $t_1$ causes a conflict for $t_2$ if and only if the disabling essence $\essdbl(t_1,t_2)$ is not trivial, i.e., if it is not the pair of morphisms $(!_{L_1}\colon\bot \to L_1, !_{L_2}\colon \bot \to L_2)$~\cite[Theorem~3.13]{ACR19}.

\paragraph{Initial conflicts.}
Starting with the notion of critical pairs, parallel (in)dependence of transformation pairs has been related to the kind of transformation pairs they \emph{extend}, culminating in the concept of \emph{initial conflicts}~\cite{LBKST19, LO21} and, for plain rules, the construction of initial conflicts from \emph{conflict reasons}~\cite[Theorem 4.4]{ACR19}.
In this paper, we continue this line of research to complete the picture for rules with application conditions. 
The following definitions briefly introduce the necessary concepts. 
For a detailed introduction to symbolic transformation pairs, the construction of the corresponding application conditions, and initial conflicts, we refer to~\cite{EGHLO12, ACR19, LO21}.

\begin{definition}[Extension of transformation. Initial transformation pair]\label{def:extension}
	Given a transformation $t = \trans{G}{\rho}{m}{H}$, an \emph{extension} of $t$ is a transformation $t' = \trans{G'}{\rho}{m'}{H'}$ such that there is an \emph{extension morphism} $f \colon G \to G'$ with $f \circ m = m'$ such that $f$ is a match in $G'$ for the transformation $t\colon G \hookleftarrow D \hookrightarrow H$ when considered as rule. 

	Given a transformation pair $(t_1, t_2)$, an \emph{initial transformation pair} for $(t_1, t_2)$ is an ac-disregarding transformation pair $(t_1^I, t_2^I)$ such that (i) there is a common extension morphism $f^I$ from the transformation $t_1^I$ to $t_1$ and from $t_2^I$ to $t_2$ and (ii) the transformation pair $(t_1^I, t_2^I)$ factors uniquely through any other transformation pair that $(t_1, t_2)$ extends by a common extension morphism.
\end{definition}

\begin{definition}[Symbolic transformation pair and initial conflict]\label{def:initial-conflict}
	Given two rules $\rho_1$ and $\rho_2$, a \emph{symbolic transformation pair} $\mathit{stp}_G = (\mathit{tp}_G, \mathit{ac}_G, \mathit{ac}^*_G)$ consists of a pair $\mathit{tp}_G\colon \transPair{H_1}{\rho_1}{m_1}{G}{\rho_2}{m_2}{H_2}$ of ac-disregarding transformations and conditions $\mathit{ac}_G, \mathit{ac}^*_G$ over $G$ 
	defined as $\mathit{ac}_G = \Shift(m_1, ac_1) \wedge \Shift(m_2, ac_2)$ and $\mathit{ac}^*_G = \neg (ac^*_{G,d_1} \wedge ac^*_{G,d_2})$, where $ac^*_{G,d_1} = \mathrm{false}$ if $d_1$ does not exist and 
	\begin{equation*}
		ac^*_{G,d_1} = \Left(G \hookleftarrow D_1 \hookrightarrow H_1, \Shift(h_1 \circ d_1, ac_2))
	\end{equation*}
	otherwise (analogously for $ac^*_{G,d_2}$).
	
	An \emph{initial conflict} for two rules $\rho_1$ and $\rho_2$ is a symbolic transformation pair $\mathit{stp}_K = (\mathit{tp}_K, \mathit{ac}_K, \mathit{ac}^*_K)$ such that
	$\mathit{tp}_K$ is (i) an initial transformation pair for itself and (ii) is \emph{conflict-inducing}, i.e., it extends to at least one conflicting pair of transformations $\transPair{H_1}{m_1}{\rho_1}{G}{m_2}{\rho_2}{H_2}$.
\end{definition}
In the construction of symbolic transformation pairs, $\mathit{ac}_G$ indicates when extending transformation pairs satisfy the rules' application conditions and $\mathit{ac}^*_G$ indicates when an extending pair of transformations is in conflict.
A transformation pair is in conflict if and only if an initial conflict $\mathit{stp}_K = (\mathit{tp}_K, \mathit{ac}_K, \mathit{ac}^*_K)$ extends to it and the extension morphism satisfies the condition $\mathit{ac}_K \wedge \mathit{ac}^*_K$. 
Given two rules, the initial conflict $(\transPair{R_1 + L_2}{i_{L_1}}{\rho_1}{L_1 + L_2}{i_{L_2}}{\rho_2}{L_1 + R_2}, ac_{L_1+L_2}, ac^*_{L_1+L_2})$, where $i_{L_1}$ and $i_{L_2}$ are obtained by the universal property of the coproduct, plays an important role in \cref{chapter:symbolic_conflict_essences} of this paper, as it is the only initial conflict that indicates a conflict caused solely by violations of the application conditions. 
We refer to it as the \emph{symbolic initial conflict}.

	\section{Construction and Embedding of Disabling Essences}
	\label{chapter:essences}
	
In this section, we extend the notion of \emph{disabling essences} 
to cover arbitrary nested application conditions and develop a concept for embedding these into transformation pairs. 
To convey some intuition, we first continue our running example 
and then develop the formal details. 

\begin{example}[Disabling essences]\label{ex:disabling_essences}
\Cref{fig:deleting_essences} shows a selection of the disabling essences for the two rules in \cref{fig:examplerules}.
 The graphs with gray borders show how the LHSs and/or application condition graphs need to overlap in a pair of transformations for a conflict to potentially occur. 
 The inner graphs with orange borders constitute the essence graphs, and the identifiers in the nodes indicate how the graphs are mapped to the LHSs and/or application condition graphs. 

\begin{figure}[t]
		\includegraphics[width = \textwidth]{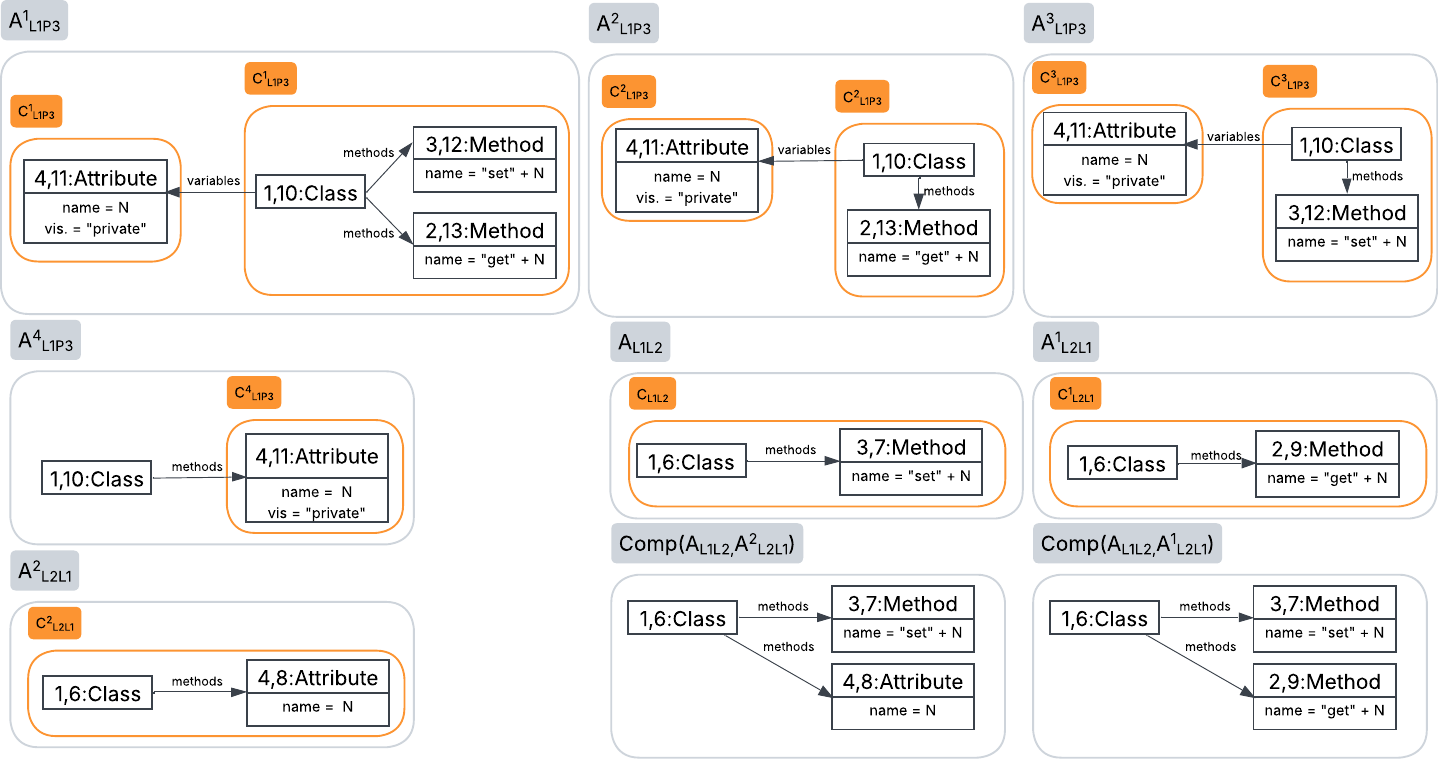}
		\vspace{-1em}
		\caption{Disabling essences of \emph{decapsulateAttribute} and \emph{pullUpEncapsulatedAttribute} and composed essences}\label{fig:deleting_essences}
		\vspace{-1.5em}
	\end{figure}

First, both rules can conflict with each other by matching to the same \textsf{Methods} and/or \textsf{Attribute}. 
The disabling essences \textsf{A\textsubscript{L1L2}, A\textsuperscript{1}\textsubscript{L2L1}}, and \textsf{A\textsuperscript{2}\textsubscript{L2L1}} can be computed using the methods developed for plain rules in~\cite{LBKST19,ACR19}. 
Importantly, when the two rules are matched as indicated by one of these essences, the resulting transformation pair is guaranteed to be conflicting. 

Second,  the \emph{decapsulateAttribute} rule can conflict with the \emph{pullUpEncapsulatedAttribute} rule by invalidating its application condition. 
To capture the core structure of these conflicts in a manner similar to that of classic disabling essences for plain rules, we expect results as displayed in the first and second row of \cref{fig:deleting_essences} (cases A\textsuperscript{1}\textsubscript{L1P3}--A\textsuperscript{4}\textsubscript{L1P3}). 
The  \emph{decapsulateAttribute} rule can invalidate \emph{pullUpEncapsulatedAttribute's} application condition by deleting
an occurrence of $P_3$ in \cref{fig:exampleapplcond}, for which there are different possibilities (e.g.  deleting the getter or setter method, setting the visibility to public, or combinations thereof. 
While not possible in this example (given the semantics of the \emph{decapsulateAttribute} rule), in principle  \emph{pullUpEncapsulatedAttribute's} application condition can also be invalidated by a rule creating an occurrence of the graph $P_1$ or deleting an occurrence of $P_2$. 
This occurs, for instance, with the rule \emph{encapsulateAttribute}, the inverse of \emph{decapsulateAttribute}, which switches the visibility of the attribute from public to private. 
\end{example}

On closer inspection of the example, it becomes apparent that, in principle, disabling essences between rules and application conditions can be found through overlaps in the same way as disabling essences between plain rules.  
The central structural difference is that these overlaps are no longer restricted to the LHSs of the rules, but also involve the application condition graphs. 
Moreover, conflicts cannot only arise from the deletion of existentially bound graphs of the application condition (as illustrated above) but also from the creation of universally bound graphs: 
A further refactoring rule could invalidate applications of \pullAttribute{} by creating occurrences of $P_1$ (newly inserting inheriting classes). 
Importantly, (i) at points of bifurcation (like the disjunction in our example) all options have to be considered  because each one can be the source of a conflict, and (ii) creating actions can also cause conflicts.
Finally, the central semantic difference from the situation with plain rules is that matching a pair of rules as indicated by a disabling essence no longer  guarantees a conflicting transformation pair. 
In our example, the deletion of an occurrence of $P_3$ may be compensated for by the presence of an occurrence of $P_2$, and vice versa. 
Moreover, even though it is semantically dubious in our specific example, there could be several occurrences of the graphs $P_2$ and $P_3$, such that deleting one of them would not invalidate the application condition.

In this section, we present the construction of disabling essences for rules with \emph{arbitrary nested application conditions}. 
We started this work in~\cite{LauerKT25}, where we computed disabling essences for application conditions in \ac{ANF}. 
This is a special case in which a condition is essentially a chain of morphisms (i.e., there are no points of bifurcation caused by conjunction or disjunction). 
The restriction in~\cite{LauerKT25} mainly served to simplify the presentation. 
When conditions are in \ac{ANF}, disabling essences can be computed using the construction developed in~\cite{ACR18, ACR19} for plain rules,  which is then applied iteratively along the chain of morphisms that constitute the condition.
For arbitrary conditions, whose morphisms form a tree structure (which we will introduce next; compare also~\cite{RensinkC25}), the situation is slightly more involved. 
The construction must be applied iteratively along each path from the root to a leaf of the tree.

\begin{definition}[Tree structure of a condition]\label{def:tree_structure_condition}
	Given a condition $c$ over an object $P_0$, the tree structure of $c$ is a graph $\Tree(c) = (V_{\Tree(c)}, E_{\Tree(c)}, \src_{\Tree(c)},\allowbreak \tar_{\Tree(c)})$ which is recursively constructed as follows: 
	If $c = \bigodot_{i \in \mathcal{I}} Q (p_1^i \colon P_0 \to P_1^i, d_i)$ for some index set $\mathcal{I}$, $\odot \in \{\vee, \wedge\}$, and $Q \in \{\exists, \forall\}$, we set $V_{\Tree(c)} = \bigcup_{i \in \mathcal{I}} V_{\Tree(d_i)} \cup \{P_0\}$, $E_{\Tree(c)} = \bigcup_{i \in \mathcal{I}} (E_{\Tree(d_i)} \cup \{p_1^i\})$, $\src_{\Tree(c)} = \bigcup_{i \in \mathcal{I}} (\src_{\Tree(d_i)} \cup\allowbreak \{(p_1^i, P_0)\})$, and $\tar_{\Tree(c)} = \bigcup_{i \in \mathcal{I}} (\tar_{\Tree(d_i)} \cup \{(p_1^i, P_1^i)\})$. 
	If $c = \true$ is a condition over  $P_i$, we set $\Tree(\true) = ({P_i}, \emptyset, \emptyset, \emptyset)$, and
	$\Tree(c) = \Tree(c')$ if $c = \neg c'$.
	Each node of $\Tree(c)$ without outgoing edges is called a \emph{leaf}, the set of all leaves is denoted by $ \Leafs(\Tree(c))$. 
\end{definition}

\begin{example}[Tree structure]\label{ex:treeStructure}
\Cref{fig:exampleapplcond} shows the 
application condition of the \emph{pullUpEncapsulatedAttribute} rule.
The tree structure of this condition has four nodes starting with the LHS of the rule as root, which is embedded into graph $P_1$, which is again embedded into $P_2$ and into $P_3$, the two leaves of the tree.
\end{example}

\begin{figure}[t]
	\centering
	\begin{tikzpicture}
		\matrix (m) [	matrix of math nodes,
		nodes in empty cells,
		row sep=.5em,
		column sep=.8em,
		minimum width=1em]
		{
			B & & C \\ 
			& \text{\scriptsize (IPO)} \\
			A'_{} & & A_{} & & P_i & & & L_2\\
			& \text{\scriptsize (PB)} & & \text{\scriptsize (PB)} & \\ 
			K_1 & & L_1 & & L_1P_n \\
		};
		\path[-stealth] 
		\Edge{1-1}{}{above}{q_i'}{1-3}
		\Cedge{}{}{}{3-1}
		\Edge{3-1}{\rmono}{fill=white}{q_i}{3-3}
		\Cedge{\lmono}{}{}{5-1}
		
		\Edge{3-3}{}{above}{b_i}{3-5}
		\Cedge{\lmono}{fill=white}{a_1}{5-3}
		\Edge{5-1}{\rmono}{fill=white}{l_1}{5-3}
		\Edge{5-3}{}{fill=white}{p_{L_1}}{5-5}
		\Edge{3-5}{\lmono}{right}{e_{P_i}}{5-5}
		\Edge{3-8}{}{above}{p_{i-1} \circ \ldots \circ p_1}{3-5}
		
		\Edge{1-3}{}{fill=white}{c}{3-3}
		;
		
	\end{tikzpicture}
	\vspace{-.3cm}
	\caption{Construction of proto-essences for rules with application conditions}\label{fig:conflict_construction}
	\vspace{-1.5em}
\end{figure}
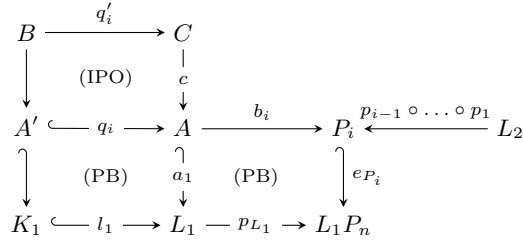

Next, we define \emph{proto-essences} for rules with application conditions; these are later filtered to obtain \emph{disabling essences}.
For their construction, we consider each path $p$ in the tree structure of the second rule's application condition that starts at the root and ends at a leaf. 
For each overlap (cospan) between the \ac{LHS} $L_1$ of the first rule and a leaf $P_n$, 
we iteratively calculate whether applying the rule can destroy an occurrence of a graph $P_i$ contained within the path. 
We achieve this by applying the computation of proto-essences of plain rules to $L_1$ and $P_i$ for all $i \le n$. 
If such an object is found, a \emph{proto-essence by deletion} has been identified. 
To account for the fact that the first rule might also invalidate the second rule's application condition by creating an occurrence of one of its objects $P_i$, we also define \emph{proto-essences by insertion}. 
These are computed analogously to the ones by deletion, just using the \ac{RHS} $R_1$ of the first rule instead of its \ac{LHS} $L_1$.

\begin{definition}[Proto-essence for rules with conditions]\label{def:essencesAC}
	Let two rules $\rho_j = (L_j \xhookleftarrow{l_j} K_j \xhookrightarrow{r_j} R_j, \mathit{ac}_j)$, with $j=1,2$, the tree-representation $\Tree(\mathit{ac}_2)$ of  $\mathit{ac}_2$, an object $P_n \in \Leafs(\Tree(\mathit{ac}_2)) \cup \{L_2\}$, a path $p_{L_2} = p_n \circ \ldots p_1$ from $L_2$ to $P_n$ in $\Tree(ac_2)$,  and a cospan $(e_{L_1} \colon L_1 \to L_1P_n, e_{P_n} \colon P_n \mono L_1P_n) \in \Epi$ be given. 
The \emph{proto-essence by deletion} of $e_{L_1}$ and $e_{P_n}$ at level $1 \leq i \leq n$, denoted as $\DPEss{i}{e_{L_1}}{e_{P_n}}$, is defined as
\begin{align*}
	\DPEss{0}{e_{L_1}}{e_{P_n}} & \coloneqq \Ess{l_1}{e_{L_1}}{e_{P_1}\circ p_1}  , \\ 
	\DPEss{i}{e_{L_1}}{e_{P_n}} & \coloneqq \begin{cases}
		\DPEss{i-1}{e_{L_1}}{e_{P_n}} &\text{if } \DPEss{i-1}{e_{L_1}}{e_{P_n}} \neq (!_{L_1}, !_{P_{i-1}}) \\
		\Ess{l_1}{e_{L_1}} {e_{P_i}} &\text{if } \DPEss{i-1}{e_{L_1}}{e_{P_n}} = (!_{L_1}, !_{P_{i-1}})
	\end{cases} ,
\end{align*}
where 
$p_n \circ \ldots \circ p_{i+1}$ is the path from $P_i$ to $P_n$ in $\Tree(\mathit{ac}_2)$ (and $P_i$ is an object contained in the path from $L_2$ to $P_n$) and $e_{P_i} := e_{P_n} \circ p_n \circ \ldots \circ p_{i+1}$ (see \cref{fig:conflict_construction}).
The \emph{set of proto-essences by deletion of $\rho_1$ in $\rho_2$}, denoted by $\DPEss{}{\rho_1}{\rho_2}$, is the set of all proto-essences by deletion at level $n$ of morphism pairs $(e_{L_1}\colon L_1 \to L_1P_n, e_{P_n}\colon P_n \mono L_1P_n) \in \Epi$ of all leaves $P_n$ contained in $\Tree(\mathit{ac}_2)$.

The set of \emph{proto-essences by insertion}, denoted by $\IPEss{}{\rho_1}{\rho_2}$ is defined as $\IPEss{}{\rho_1}{\rho_2} \coloneq \DPEss{}{\rho^{-1}_1}{\rho_2}$ where we stop the recursion at $i=1$.

\end{definition}

\begin{example}[Construction of a proto-essence by deletion]\label{ex:construction_deleting_essence}
	In this example, we sketch the computation of a proto-essence by deletion between $L_1$ of \decapAttribute{} and the path $L_2, p_1, P_1, p_3,P_3$ in the application condition of \pullAttribute{}; a detailed version is provided as \cref{ex:construction_deleting_essence_complete} in \cref{app:detailed_examples}.
	\Cref{fig:ex_construction_proto_essence_compact1} shows the crucial part of the computation, not displaying graphs that are already known or identical to other graphs displayed. 

	\begin{figure}[t]
		\includegraphics[width = \textwidth]{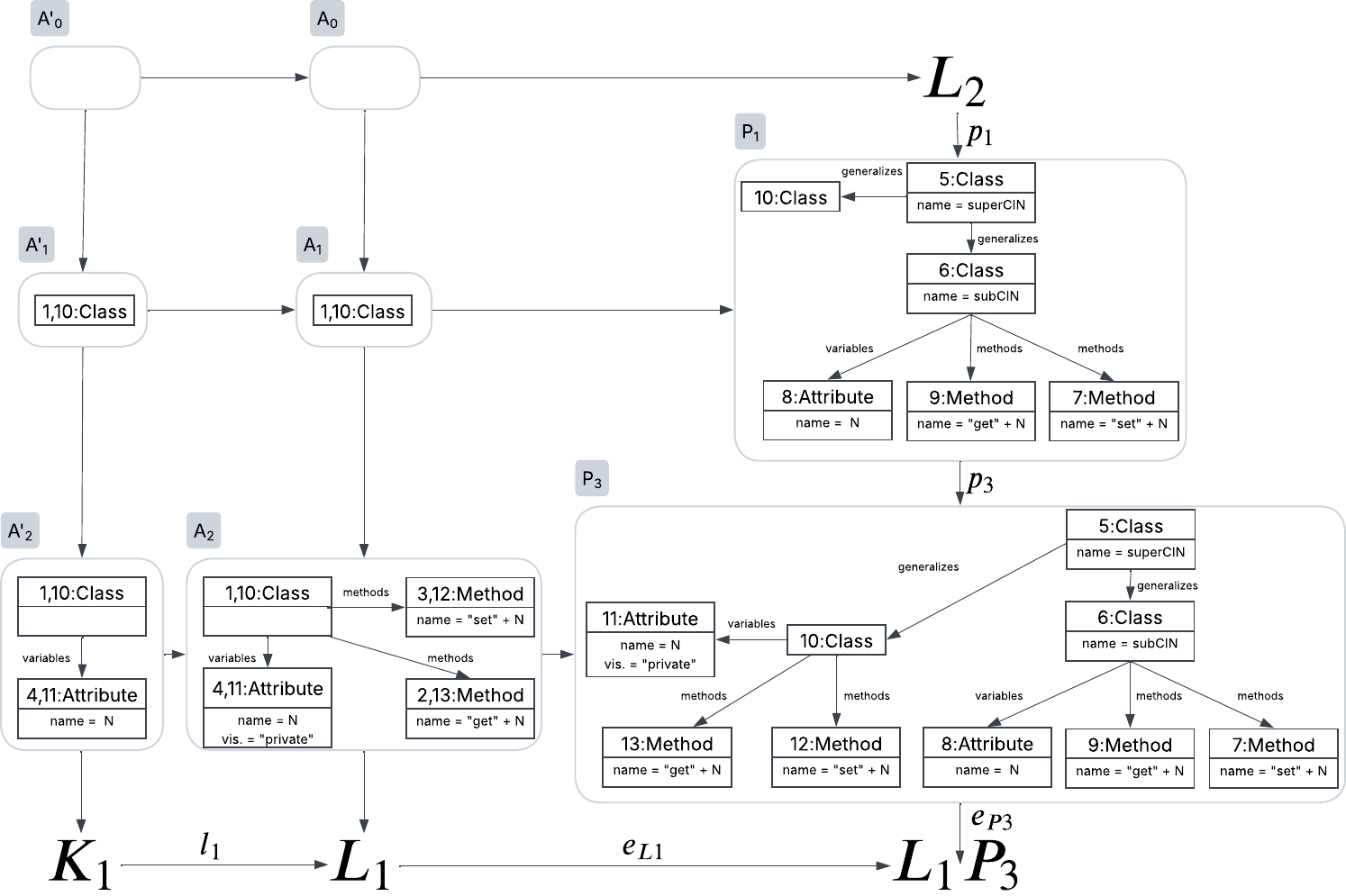}
		\vspace{-1.5em}
		\caption{Central extract of the construction of proto-essences by deletion of the rules \decapAttribute\ and \pullAttribute }\label{fig:ex_construction_proto_essence_compact1}
	\end{figure}
	
	The construction starts with the overlap $L_1P_3 = P_3$ of $L_1$ and $P_3$, to which $P_3$ is mapped via the identity morphism $e_{P_3} = \mathit{id}_{P_3}$ and $L_1$ by the morphism that is fixed by mapping \textsf{1:Class} to \textsf{10:Class}. 
	It proceeds by computing the proto-essence by deletion at level 0, which is done by computing the pullback for $L_1$ and $L_2$ in $L_1P_3$, resulting in $A_0$ (the empty graph), followed by computing the pullback of $A_0$ and $K_1$ in $L_1$, resulting in $A_0'$ (also the empty graph). 
	That $A_0$ and $A_0'$ are isomorphic indicates that the proto-essence at this level will be trivial; we therefore continue to further levels, which are computed analogously. 
	The computation at level 1 is done using $P_1$ instead of $L_2$, leading to isomorphic graphs $A_1$ and $A_1'$, again indicative of a trivial proto-essence at this level. 
	At level 2, we then use $P_3$ instead of $L_2$ and compute $A_2$ and $A_2'$, which are not isomorphic. 
	Because the proto-essences computed at lower levels were trivial, the proto-essence by deletion at level 2 turns out to be \emph{the} proto-essence by deletion for the considered overlap $L_1P_3$.
	
	The actual essence graph is obtained by computing the initial pushout over the morphism from $A_2'$ to $A_2$. 
	It indicates as potential cause of a conflict that \decapAttribute{} simultaneously deletes the getter and setter methods and disables the attribute's visibility necessary for \pullAttribute{}'s application condition to be valid. 
	It is displayed as \textsf{A\textsuperscript{1}\textsubscript{L1P3}} in \cref{fig:deleting_essences}. 
\end{example}

To define \emph{disabling essences}, we need a notion of \emph{embeddability} for proto-essences. 
To be able to provide a single definition of \emph{embeddability} that also covers the further kinds of essences we introduce in this paper, we first define \emph{rule overlaps} (capturing the structure all essences have) and then define embeddings for rule overlaps.

\begin{definition}[Rule overlap. Embedding of rule overlap]\label{def:embedding:spans}
	Given two rules $\rho_1 = \rleAC{1}$, $\rho_2 = \rleAC{2}$, an \emph{overlap} of $\rho_1$ and $\rho_2$ is a tuple $\mathit{ro} = (b_j \colon A \to P_j, b_i \colon A \to P_i, p_{L_1} \colon L_1 \to P_j, p_{L_2} \colon L_2 \to P_i)$. 
	Here, $P_j = L_1$ and/or $P_i = L_2$ is allowed. 
	A pair of morphisms $(m_1^j \colon P_j \to G, m_2^i \colon P_j \to G)$ is an \emph{embedding of $\mathit{ro}$ in a transformation pair $(t_1,t_2) = (\trans{G}{\rho_1}{m_1}{H_1}, \trans{G}{\rho_2}{m_2}{H_2})$} if $m_1 = m_1^j \circ p_{L_1}$, $m_2 = m_2^i \circ p_{L_2}$ and the square (1) in \cref{fig:embedding_disabling_essence} is a pullback.

		\begin{figure}[t]
			\centering
		\begin{tikzpicture}
			\matrix (m) [	matrix of math nodes,
			nodes in empty cells,
			row sep=1em,
			column sep=.8em,
			minimum width=1em]
			{
				& &     & &     & C&     & A      &     &  & L_2'    & & K_2'    & & \\
				& &     & &     & &     &        &     &\text{\scriptsize (2)} &     &\text{\scriptsize (3)} &     & & \\
				R_1 & & K_1 & & L_1 & & P_j/L_1 &    \text{\scriptsize(1)}    & P_i & & L_2 & & K_2 & & R_2 \\				
				& &     & &     & &     &        &     & &     & &     & & \\
				H_1	& & D_1 & &     & &     & G &     & &     & & D_2 & & H_2 \\
			};
			
			\path[-stealth]
			
			(m-1-6)edge node[fill=white]{\edge{c}}(m-1-8)
			
			(m-1-8) edge[] node [fill=white] {\edge{b_j}} (m-3-7)
			(m-1-8) edge[] node [fill=white] {\edge{b_i}} (m-3-9)
			
			(m-3-3) edge[\lmono] node [fill = white] {\edge{r_1}} (m-3-1)
			edge[\rmono] node [fill = white] {\edge{l_1}} (m-3-5)
			edge[] node [] {} (m-5-3)
			
			(m-3-1)edge[] node [fill = white] {\edge{n_1}} (m-5-1)
			
			(m-5-3)edge[\lmono] node [fill = white] {\edge{h_1}} (m-5-1)
			edge[\rmono] node [fill = white] {\edge{g_1}} (m-5-8)
			
			(m-3-5)edge[] node [above] {\edge{p_{L_1}/ id_{L_1}}} (m-3-7)
			edge[bend right = 10] node [fill = white,near end] {\edge{m_1}} (m-5-8)
			
			(m-3-7)edge[] node [fill = white] {\edge{m_1^j}} (m-5-8)
			
			(m-3-9)edge[] node [fill = white] {\edge{m_2^i}} (m-5-8)
			
			(m-3-11)edge[] node [fill = white] {\edge{p_{L_2}}} (m-3-9)
			edge[bend left = 10] node [fill = white,near end] {\edge{m_2}} (m-5-8)
			
			(m-3-13)edge[\lmono] node [fill=white]{\edge{l_2}}(m-3-11)
			edge[\rmono] node [fill=white]{\edge{r_2}}(m-3-15)
			edge[] node []{}(m-5-13)
			
			(m-5-13) edge[\lmono] node [fill=white] {\edge{g_1}}(m-5-8)
			edge[\rmono] node [fill=white] {\edge{h_1}}(m-5-15)
			
			(m-3-15)edge[] node [fill=white] {\edge{n_2}}(m-5-15)
			
			\Edge{1-11}{}{}{}{1-8}
			\Cedge{}{}{}{3-11}
			
			\Edge{1-13}{\lmono}{fill=white}{l_2'}{1-11}
			\Cedge{}{}{}{3-13}
			
			;

		\end{tikzpicture}
		\vspace{-2em}
		\caption{Diagram for the embedding of rule overlaps and ac-conflicting disabling essences. Consider the alternatives $P_j$ for the embedding of rule overlaps and $L_1$ (separated by ``/'') for ac-conflicting disabling essences}\label{fig:embedding_disabling_essence}
		\vspace{-1.5em}
	\end{figure}
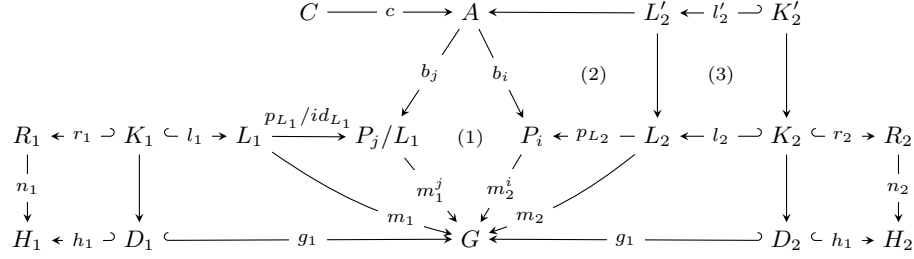 
	
\end{definition}

\emph{Disabling essences for rules with application conditions} can now be defined analogously to~\cite{LauerKT25}. 
For reasons of space, the exact construction will not be presented in this section, as it involves a somewhat lengthier construction to \emph{shift} proto-essences by insertion so that both kinds of essences are defined as the same kinds of spans. 
Intuitively, the definition of disabling essences filters out proto-essences that can never indicate a conflict by requiring that a disabling essence (i) is non-trivial, (ii) can actually be embedded in a transformation pair, and (iii) indicates the creation of a universally bound object or the deletion of an existentially bound one (but not vice versa); this final requirement is well-defined if conditions are given as assumed here, i.e., if all negations are pushed inwards as far as possible. 
Importantly, given the more general definition of proto-essences here, the requirements for disabling essences remain identical to those presented in~\cite{LauerKT25}. 
Moreover, each result presented in~\cite{LauerKT25} directly transfers to our extended construction. 
We have included the concrete definition of \emph{disabling essences} and re-proved the main results of~\cite{LauerKT25} in Appendix~\ref{sec:previous_results}.

	\section{From Disabling Essences to Conflict Essences}
	\label{chapter:embedding}

The disabling essences constructed in the previous section are ``asymmetric''  in that they only indicate potential conflicts in one direction. 
Next, we 
introduce \emph{conflict essences for rules with application conditions} that capture potential conflicts between rules in both directions simultaneously. 
These conflict essences are defined by 
composing the disabling essences for both directions.

When we start to compose disabling essences, i.e.,  when we consider the effects of both transformations on each other, some of the disabling essences that we previously computed become obsolete: 
A disabling essence $de = (b_1 \circ c \colon C \to L_1, b_i \circ c \colon C \to P_i) \in \DEss{\rho_1}{\rho_2}$ indicates that the application of $\rho_1$ can prevent the application of $\rho_2$ by invalidating its application condition (by destroying an occurrence of $P_i$). 
However, when the roles are switched, the overlap used to compute $de$ may instead produce a disabling essence $de' = (b_2 \circ c' \colon C \to L_2, b_1' \circ c' \colon C' \to L_1) \in \DEss{\rho_2}{\rho_1}$, which indicates that the application of $\rho_2$ can prevent the application of $\rho_1$ by destroying its match. 
In other words, when this disabling essence is embedded in a pair of transformations,  
the underlying ac-disregarding transformations are already in conflict. 
We introduce \emph{ac-conflicting rule overlaps} which imply ac-dis\-re\-gard\-ing parallel independence (i.e., neither transformation destroys the match of the other) when embedded in a transformation pair. 
Once again, we use \emph{rule overlaps} to ensure that the definition can be applied to any kind of essence considered in this paper.

\begin{definition}[Ac-conflicting rule overlap]\label{def:ac-conflicting}
	Given rules $\rho_1$ and $\rho_2$, a rule overlap $\mathit{ro}$ of them is called \emph{ac-conflicting} if there is an ac-disregarding parallel independent transformation pair $(t_1,t_2)$ via $\rho_1$ and $\rho_2$ into which $\mathit{ro}$ embeds.
\end{definition}

The following result shows that each transformation pair in which an ac-conflicting rule overlap is embedded is parallel independent ac-disregarding. 

\begin{proposition}[Disjoint sets of ac-conflicting and non-ac-conflicting rule overlaps]\label{thm:disjointness_ac-disregarding}
	If a rule overlap $\mathit{ro}$ 
	is \emph{ac-conflicting}, then each transformation pair $(t_1,t_2)$ such that $\mathit{ro}$ is embeddable in $(t_1,t_2)$ is parallel independent ac-disregarding. 
	Moreover, if a non-ac-conflicting rule overlap is embeddable in a transformation pair $(t_1,t_2)$, this pair is parallel dependent ac-disregarding.
\end{proposition}

The following criterion can be used to decide whether a disabling essence $(b_1 \colon A \to L_1, b_j \colon A \to P_j)$ of a rule pair is ac-conflicting;  
we illustrate its use in \cref{ex:ac-conflicting-essence} in \cref{app:detailed_examples}.

\begin{proposition}[Ac-conflicting disabling essences]\label{thm:embedding_of_ac_conflicting_essences}
	Given  
	rules $\rho_1$ and $\rho_2$, a disabling essence $ce = (b_1 \colon A \to L_1, b_i \colon A \to P_i, id_{L_1}, p_{L_2}) \in \ESS{\rho_1}{\rho_2}$ is \emph{ac-conflicting} if and only if the morphism $l_2'$, which is obtained by computing the pullbacks (2) and (3) shown in \cref{fig:embedding_disabling_essence}, is an isomorphism. 
\end{proposition}

We now develop a construction for composing rule overlaps (encompassing disabling essences) in such a way that composed overlaps can be embedded in a transformation pair if the individual overlaps are embedded in it, and vice versa. 

\begin{definition}[Composition of rule overlaps]\label{def:composition_overlaps}
	Given two rule overlaps $\mathit{ro} = (\Span{P_i}{b_i}{A}{b_j}{P_j}, e_{L_1}, e_{L_2})$ and $\mathit{ro}' = (\Span{P'_i}{b'_i}{A'}{b'_j}{P'_j}, e'_{L_1}, e'_{L_2})$, 
	the \emph{set of compositions} of $\mathit{ro}$ and $\mathit{ro}'$, denoted by $\Comp{\mathit{ro}}{\mathit{ro}'}$, contains each rule overlap 
	$(\MonoSpan{P_j^*}{b_j^*}{A^*}{b_i^*}{P_i^*},e_{P_j} \circ p_{L_1},e_{P_i} \circ p_{L_2} )$ so that $(e_{P_j}\colon P_j \to P_j^*, e_{P'_j} \colon P_j' \to P_j^*)$, $(e_{P_i} \colon P_i \to P_i^*, e_{P'_i} \colon P_i' \to P_i^*)$, $(e_{P^*_j} \colon\allowbreak P_j^* \mono K^*, e_{P^*_i} \colon P_j^* \mono K^*) \in \Epi$, $e_{P_j} \circ p_{L_1} = e_{P_j'} \circ p_{L_1}'$, $e_{P_i} \circ p_{L_2} = e_{P_i'} \circ p_{L_2}'$, the outer squares in the diagrams shown in \cref{fig:comp} are pullbacks and $b^*_j$ and $b_i^*$ are obtained by computing the pullback (1).
	
	\begin{figure}[t]
		\centering
		\begin{tikzpicture}
			\matrix (m) [	matrix of math nodes,
			nodes in empty cells,
			row sep=1em,
			column sep=1em,
			minimum width=1em]
			{
				P_i & & &  A &  A' & &  & P_i'\\ 
				P_i^* & & A^*_{} & & &    A^*_{} & & P_i^*\\ 
				& \text{\scriptsize(1)} &    & & & &\text{\scriptsize(1)}\\
				K^*_{} & & P_j^*  & P_j  & P_j' &  P_j^* & & K^* \\
			};
			
			\path[-stealth] 
			\Edge{1-4}{}{fill=white}{b_i}{1-1}
			\Cedge{}{fill=white}{b_j}{4-4}
			
			\Edge{1-5}{}{fill=white}{b_i'}{1-8}
			\Cedge{}{fill=white}{b_j'}{4-5}
			
			\Edge{2-6}{\rmono}{fill=white}{b_i^*}{2-8}
			\Cedge{\rmono}{fill=white}{b_j^*}{4-6}
			
			\Edge{2-3}{\lmono}{fill=white}{b_i^*}{2-1}
			\Cedge{\rmono}{fill=white}{b_j^*}{4-3}
			
			\Edge{1-1}{}{left}{e_{P_i}}{2-1}
			\Edge{2-1}{\rmono}{fill=white}{e_{P_i^*}}{4-1}
			
			\Edge{2-8}{\rmono}{fill=white}{e_{P^*_i}}{4-8}
			
			\Edge{1-8}{}{right}{e_{P_i'}}{2-8}
			
			\Edge{4-3}{\lmono}{below}{e_{P_j^*}}{4-1}
			\Edge{4-6}{\rmono}{fill=white}{e_{P^*_j}}{4-8}
			
			\Edge{4-4}{}{below}{e_{P_j}}{4-3}
			\Edge{4-5}{}{below}{e_{P_j'}}{4-6}
			
			;
		\end{tikzpicture}
		\vspace{-1em}
		\caption{Diagram for the composition of rule overlaps } \label{fig:comp}
		\vspace{-1.5em}
	\end{figure}
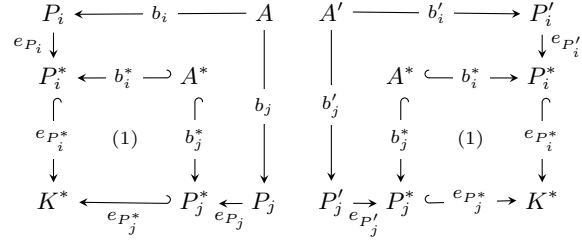
	
\end{definition}

The following example briefly illustrates the result of composing rule overlaps; we illustrate a computation in detail as \cref{ex:composedRuleOverlap1} in \cref{app:detailed_examples}. 
\begin{example}[Composed rule overlap]
\label{ex:composedRuleOverlap}
	The composed rule overlap \textsf{Comp($A_{L1L2},\allowbreak A^1_{L2L1}$)} in \cref{fig:deleting_essences} indicates that $A_{L1L2}$ and $A^1_{L2L1}$ can be embedded in the same transformation by identifying \textsf{1,6:Class}. 
	An occurrence of this composed overlap simultaneously indicates (i) that \emph{decapsulateAttribute} causes a conflict with \emph{pullUpEncapsulatedAttribute} by destroying \textsf{3,7:Method} 
	and (ii)
	that \emph{pullUpEncapsulatedAttribute} causes a conflict with \emph{decapsulateAttribute} by destroying the edge leading to \textsf{2,9: Method}.
\end{example}

 We show that the desired property of composed rule overlaps is satisfied, i.e., that a composed overlap is embedded in a transformation pair if and only if both of its constituent rule overlaps are also embedded in the pair. 

\begin{theorem}[Embedding of composed rule overlaps]\label{thm:embedding_conflict_essences}
	Given two rule overlaps $\mathit{ro}$ and $\mathit{ro}'$ of two rules $\rho_1$, and $\rho_2$, 
	if both rule overlaps $\mathit{ro}$ and $\mathit{ro}'$ are embeddable in a transformation pair $(t_1,t_2)$ via $\rho_1$ and $\rho_2$, then there is a composed overlap $co \in \Comp{\mathit{ro}}{\mathit{ro}'}$ that is also embeddable in $(t_1,t_2)$.
	Conversely, if any composed overlap $co \in \Comp{\mathit{ro}}{\mathit{ro}'}$ is embedded in a transformation pair $(t_1,t_2)$ via $\rho_1$ and $\rho_2$, then both $\mathit{ro}$ and $\mathit{ro}'$ are also embedded in $(t_1,t_2)$. 
\end{theorem}

A crucial observation is that if both rule overlaps are ac-conflicting, then so are all composed overlaps.
This statement follows directly from the second part of \cref{thm:embedding_conflict_essences}.
We will use this result later on when comparing conflict essences (which encompass composed disabling essences) 
with initial conflicts for rules with application conditions.
\begin{corollary}[Inheritance of ac-conflicting rule overlaps]
	Given two ac-conflicting rule overlaps $\mathit{ro}$ and $\mathit{ro}'$ of rules $\rho_1$ and $\rho_2$, each composed overlap $ce \in \Comp{\mathit{ro}}{\mathit{ro}'}$ is ac-conflicting.
\end{corollary}

We now define \emph{conflict essences}, which provide an overview on how two transformations can cause a conflict for each other.
To construct conflict essences, we compose rule overlaps, in which the first transformation causes a conflict with the second and the second transformation causes a conflict with the first. 
All disabling essences from \cref{fig:deleting_essences} and, in particular, the composed rule overlap from \cref{ex:composedRuleOverlap} are conflict essences. 

\begin{definition}[Set of conflict essences]
	Given two rules $\rho_j = \rleACCompact{j}$, with $j=1,2$, the \emph{set of conflict essences of $\rho_1$ and $\rho_2$}, denoted by $\CEss{\rho_1}{\rho_2}$, is recursively defined as follows: 
	\begin{enumerate}
		\item A disabling essence $de = (b_i \colon A \to P_i, b_j \colon A \to P_j) \in \ESS{\rho_1}{\rho_2} \cup \ESS{\rho_2}{\rho_1}$ is a \emph{conflict essence} for $\rho_1$ and $\rho_2$ if and only if (i) $de$ is ac-conflicting or (ii) $P_i = L_1$ and $P_j = L_2$.
		\item If $ce_1$ and $ce_2$ are conflict essences so that both are either ac-conflicting or  non-ac-conflicting, then every composed rule overlap $ce^* \in \Comp{ce_1}{ce_2}$ is a conflict essence.
	\end{enumerate}
\end{definition}

The following corollary directly follows from the fact that, for each parallel dependent transformation pair, there is a disabling essence that embeds in it~\cite{LauerKT25}. This is because the set of conflict essences contains each ac-conflicting disabling essence in both directions. 
For ac-disregarding parallel dependent transformations, the claim follows from the fact that the set of conflict essences contains each conflict essence for plain rules, together with the results shown in~\cite{ACR19}.
\begin{corollary}[Embedding of conflict essences]\label{cor:embedding_conflict_essences}
	Given a pair of parallel dependent transformations $(t_1,t_2) = (\trans{G}{\rho_1}{m_1}{H_1}, \trans{G}{\rho_2}{m_2}{H_2})$ via rules 
	$\rho_1$ and $\rho_2$, there is a conflict essence $ce \in \CEss{\rho_1}{\rho_2}$ that is embeddable in $(t_1,t_2)$.
\end{corollary}

	\section{Symbolic Conflict Essences and Initial Conflicts}
	\label{chapter:symbolic_conflict_essences}

In this section, we introduce \emph{symbolic conflict essences}, which consist of conflict essences accompanied by a nested condition. 
As with initial conflicts, the embedding morphisms of a conflict essence satisfy the corresponding condition if and only if the transformation pair is parallel dependent. 
This ensures that a \emph{symbolic conflict essence} is embedded in a transformation pair if and only if the transformation pair is parallel dependent. 
Moreover, each embedded symbolic conflict essence indicates a location at which an occurrence of an application condition  object is destroyed (or created),  or a rule match is destroyed. 
We relate each symbolic conflict essence 
to a unique initial conflict. 
We then use the condition of the initial conflict to construct the condition of the conflict essence using the well-known shift operator for nested conditions~\cite{HabelP09}.
 
To ensure that the condition of a \emph{symbolic conflict essence} is satisfied if and only if the transformation pair is parallel dependent,  both embedding morphisms of the conflict essence embedding must satisfy the corresponding condition. 
To achieve this, we introduce \emph{cospan conditions}, which establish a way to restrict the satisfaction of conditions. 
Rather than using a single morphism, we will use a cospan on the first nesting level of the condition (whereas the condition itself is a nested condition as usual).
\begin{definition}[cospan condition]
	Given two objects $P_0$ and $P_0'$, $cc = \exists(P_0 \xrightarrow{p_1} P_1 \xleftarrow{p_1'} P_0',d)$ is a \emph{cospan condition over $P_0$ and $P_0'$} if $d$ is a condition over $P_1$. 
	Morever, every Boolean combination of cospan conditions over $P_0$ and $P_0'$ is a cospan condition over $P_0$ and $P_0'$. 
	
	Two morphisms $q \colon P_0 \to G$ and $q' \colon P_0' \to G$ satisfy the cospan condition $cc$, denoted by $(q,q') \models cc$, if there is a morphism $q_1 \colon P_1 \to G$ so that $q_1 \models d$, $q = q_1 \circ p_1$ and $q' = q_1 \circ p_1'$.
	
\end{definition}

Next, we introduce \emph{symbolic conflict essences}. 
To relate each \emph{conflict essence} to an initial conflict, we must consider two different cases:
Either the conflict essence is ac-conflicting, or it is not. 
For a non-ac-conflicting conflict essence, there is no need of an additional condition to check for parallel dependence, since an embedding of a non-ac-conflicting conflict essence already shows (ac-disregarding) parallel dependence (\cref{thm:disjointness_ac-disregarding}). 
Therefore, the condition of ac-conflicting conflict essences is $\true$. 

If an ac-conflicting conflict essence is embeddable in a transformation pair,  the pair is parallel independent and ac-disregarding (\cref{thm:disjointness_ac-disregarding}). 
This means that the symbolic initial conflict is also embeddable in this transformation pair~\cite[Lemma 4]{LO21}, and we can use the condition of the symbolic initial conflict to construct the condition for the conflict essence. 
If the symbolic initial conflict is embedded in a transformation pair, the transformation pair is parallel dependent ac-disregarding. 
Although each symbolic conflict essence uniquely corresponds to an initial conflict, the conflict essence is much more informative because it compactly encodes the root cause(s) of the conflicts present in the transformation pairs in which it is embedded. 

\begin{definition}[Symbolic conflict essence]\label{def:symbolic_conflict_essences}
	Given two rules $\rho_j = \rleACCompact{j}$, with $j = 1,2$, and a conflict essence $ce = (b_i \circ c \colon C \to P_i, b_j \circ c \colon C \to P_j) \in \ESS{\rho_1}{\rho_2}$, the  pair $(ce, ac_{ce})$ is a \emph{symbolic conflict essence} where
	$$ac_{ce} \coloneq \exists(\CoSpan{P_i}{b_j'}{D}{b_i'}{P_j}, \Shift(m^*, ac_{L_1+L_2} \wedge ac^*_{L_1+L_2}))$$ if $\mathit{ce}$ is ac-conflicting, and $ac_{ce} \coloneq \true$ otherwise.
	The morphism  $b_j'$ is obtained by computing the pushout (1) in \Cref{fig:construction_conflict_essence}, $m^*\colon L_1 +L_2 \to D$ is the morphism obtained by the universal property of the coproduct $L_1 + L_2$ 
	and $(\transPair{R_1 + L_2}{i_{L_1}}{\rho_1}{L_1 + L_2}{i_{L_2}}{\rho_2}{L_1 + R_2}, ac_{L_1+L_2}, ac^*_{L_1+L_2})$ is the 
	symbolic intitial conflict.
	
	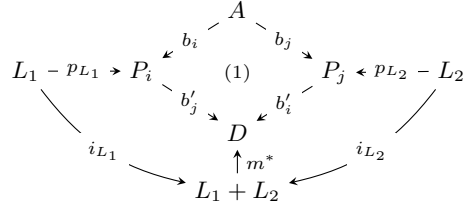
\begin{figure}[t]
		\vspace{-1em}
		\centering
		\begin{tikzpicture}
				\matrix (m) [	matrix of math nodes,
			nodes in empty cells,
			row sep=1em,
			column sep=1em,
			minimum width=1em]
			{
				& & & A_{} \\
				L_1 & & P_i & \text{\scriptsize (1)}& P_j & & L_2 \\ 
				& & & D \\ 
				& & & L_1+L_2 \\
			};
			\path[-stealth]
			\Edge{1-4}{}{fill=white}{b_j}{2-5}
			\Cedge{}{fill=white}{b_i}{2-3}
			
			\Edge{2-1}{}{fill=white}{p_{L_1}}{2-3}
			\Cedge{bend right=20}{fill=white}{i_{L_1}}{4-4}
			
			\Edge{2-7}{}{fill=white}{p_{L_2}}{2-5}
			\Cedge{bend left=20}{fill=white}{i_{L_2}}{4-4}
			
			\Edge{4-4}{}{right}{m^*}{3-4}
			
			\Edge{2-3}{}{fill=white}{b_j'}{3-4}
			\Edge{2-5}{}{fill=white}{b_i'}{3-4}
			;
		\end{tikzpicture}
		\vspace{-1em}
		\caption{Construction of symbolic conflict essences}\label{fig:construction_conflict_essence}
		\vspace{-1.5em}
	\end{figure}
		The set of \emph{symbolic conflict essences} of the rules $\rho_1$ and $\rho_2$, denoted by $\CEssAC{\rho_1}{\rho_2}$, is defined as 
		$\CEssAC{\rho_1}{\rho_2} \coloneq \{(ce, ac_{ce}) \mid ce \in \CEss{\rho_1}{\rho_2}\}.$

\end{definition}

\begin{remark}
	The pushout $(1)$ in \cref{fig:construction_conflict_essence} exists if the considered conflict essence is ac-conflicting because either $b_i$ or $b_j$ is an $\Mono$-morphism: 
	If the essence is composed, both morphisms are $\Mono$-morphisms by construction (see \cref{def:composition_overlaps}). 
	Moreover, for disabling essences of the form $(b_1 \colon A \to L_1, b_j \colon A \to P_j)$ where $P_j \neq L_2$, we have $b_1 \in \Mono$ by construction (\cref{def:asynchronous_essences}).  Disabling essences of the form $(b_1 \colon A \to L_1, b_j \colon A \to L_1)$ are also disabling essences for plain rules~\cite[Proposition 3]{LauerKT25}, which directly implies that they are not ac-conflicting.
\end{remark}

The following example briefly illustrates symbolic conflict essences; we illustrate a symbolic conflict essence in detail as \cref{ex:detailed_symbolic_essence} in \cref{app:detailed_examples}. 
\begin{example}[Symbolic conflict essence]
	Together with a cospan condition $ac_{ce}$, each conflict essence shown in \cref{fig:deleting_essences} forms a symbolic conflict essence.
	For \textsf{A\textsubscript{L1L2}, A\textsuperscript{1}\textsubscript{L2L1}}, \textsf{A\textsuperscript{2}\textsubscript{L2L1}}, and the composed essences,
	which are not ac-conflicting, $ac_{ce}$ is equal to $\true$. 
	For the remaining ones, $ac_{ce}$ is non-trivial and 
	ensures that (i) both rules satisfy their application conditions before the transformations
	and that (ii) the transformation pair is parallel dependent. 
\end{example}

Because symbolic conflict essences consist of a conflict essence and a condition, we need to extend the notion of embeddability to also account for the condition: the embedding morphisms must satisfy it.

\begin{definition}[Embeddability of symbolic conflict essences]
	Given a pair of transformations $(t_1,t_2) = (\trans{G}{\rho_1}{m_1}{H_1}, \trans{G}{\rho_2}{m_2}{H_2})$ via rules $\rho_1$ 
	and $\rho_2$, 
	a symbolic conflict essence $(ce, ac_{ce}) \in \CEssAC{\rho_1}{\rho_2}$ is embeddable in $(t_1,t_2)$ if $ce$ is embeddable in $(t_1,t_2)$ via embedding morphisms $m_1^j$ and $m_2^i$ and $(m_1^j, m_2^i) \models ac_{ce}$.
\end{definition}

We can now establish the main results of this paper, namely that symbolic conflict essences play the same role for rules with application conditions that conflict essences do for plain rules (compare \cref{fig:overview}): 
They unambiguously indicate a transformation pair to be in conflict and correspond to (symbolic) initial conflicts.
\begin{theorem}[Correctness of symbolic conflict essences]\label{thm:correctness_ce_conditions}
	Two transformations $(t_1,t_2) = (\trans{G}{\rho_1}{m_1}{H_1}, \trans{G}{\rho_2}{m_2}{H_2})$ via rules $\rho_1$ and $\rho_2$ are parallel dependent if and only if a symbolic conflict essence $(ce, ac_{ce}) \in \CEssAC{\rho_1}{\rho_2}$ is embeddable in $(t_1,t_2)$.
\end{theorem}

\begin{theorem}[Relation of symbolic conflict essences and initial conflicts]\label{relation_to_initial_conflicts}
	For each symbolic conflict essence $ce$, there is a unique initial conflict so that the initial conflict is embeddable in a transformation pair via an extension diagram if $ce$ is embeddable in the transformation pair. 
	Moreover, when an initial conflict is embedded in a transformation pair, there is at least one symbolic conflict essence that is also embedded in this transformation pair.
\end{theorem}

	\section{Conclusion}
	\label{chapter:conclusion}

In this paper, we continue our work on providing techniques to statically analyse the interactions of (graph) transformation rules. 
Concretely, we introduce \emph{conflict essences} for rules equipped with application conditions; these objects encapsulate the essential information about how two rule applications can cause a conflict with each other, potentially by invalidating one another's application condition. 
This closes an important gap because the construction was previously only available for plain rules (i.e., rules without application conditions). 
Our construction builds on the existing one, enabling us to lift the most important results from the ``plain situation'' to the case of application conditions. 
To achieve this, we construct \emph{symbolic conflict essences} (which are conflict essences equipped with an application condition), which uniquely correspond to a \emph{symbolic initial conflict}~\cite{LO21} and embed in a transformation pair if and only if that pair is conflicting. 

For a conflicting transformation pair, the set of embedded symbolic conflict essences is not minimal, in the sense that each embedded symbolic conflict essence signals a spot that causes the conflict. 
In the future, we plan to further investigate our composition operation. For example, we aim to identify the ``maximal'' ones that cannot be extended by additional spots.  
Moreover we intend to use this analysis to investigate the effects of rules on constraints to refine the graph repair approaches presented in~\cite{LauerKT24,FritscheLKST26} and to identify effective mutation orders in model-driven optimisation~\cite{JohnKLT23}.

\subsubsection{\ackname}
This work was partially funded by the German Research Foundation (DFG), project “Triple Graph Grammars (TGG) 3.0” and  ``Model-Driven Optimization in Software Engineering''.
The authors would like to thank the ICGT reviewers for their insightful comments.

\subsubsection{\discintname}

The authors have no competing interests to declare that are relevant to the content of this article.

	%
	%
	%
	\bibliographystyle{splncs04}
	\bibliography{bibliography}

@article{ACR19,
  author       = {Guilherme Grochau Azzi and
                  Andrea Corradini and
                  Leila Ribeiro},
  title        = {On the essence and initiality of conflicts in $\mathcal{M}$-adhesive transformation
                  systems},
  journal      = {J. Log. Algebraic Methods Program.},
  volume       = {109},
  year         = {2019},
  doi          = {10.1016/J.JLAMP.2019.100482}
}

@inproceedings{ACR18,
  author       = {Guilherme Grochau Azzi and
                  Andrea Corradini and
                  Leila Ribeiro},
  editor       = {Leen Lambers and
                  Jens H. Weber},
  title        = {On the Essence and Initiality of Conflicts},
  booktitle    = {Graph Transformation -- 11th International Conference, {ICGT} 2018,
                  Held as Part of {STAF} 2018, Toulouse, France, June 25--26, 2018, Proceedings},
  series       = {Lecture Notes in Computer Science},
  volume       = {10887},
  pages        = {99--117},
  publisher    = {Springer},
  year         = {2018},
  doi          = {10.1007/978-3-319-92991-0\_7}
}

@book{EEPT06,
 author       = {Hartmut Ehrig and
                  Karsten Ehrig and
                  Ulrike Prange and
                  Gabriele Taentzer},
  title        = {Fundamentals of Algebraic Graph Transformation},
  series       = {Monographs in Theoretical Computer Science. An {EATCS} Series},
  publisher    = {Springer},
  year         = {2006},
  doi          = {10.1007/3-540-31188-2},
  isbn         = {978-3-540-31187-4}
}

@book{EEGH15,
  author       = {Hartmut Ehrig and
                  Claudia Ermel and
                  Ulrike Golas and
                  Frank Hermann},
  title        = {Graph and Model Transformation -- General Framework and Applications},
  series       = {Monographs in Theoretical Computer Science. An {EATCS} Series},
  publisher    = {Springer},
  year         = {2015},
  doi          = {10.1007/978-3-662-47980-3}
}

@article{EGHLO12,
  author       = {Hartmut Ehrig and
                  Ulrike Golas and
                  Annegret Habel and
                  Leen Lambers and
                  Fernando Orejas},
  title        = {{{\(\mathcal{M}\)}-Adhesive Transformation Systems with Nested Application
                  Conditions. Part 2: Embedding, Critical Pairs and Local Confluence}},
  journal      = {Fundam. Informaticae},
  volume       = {118},
  number       = {1--2},
  pages        = {35--63},
  year         = {2012},
  doi          = {10.3233/FI-2012-705}
}

@inproceedings{EHPP04,
  author       = {Hartmut Ehrig and
                  Annegret Habel and
                  Julia Padberg and
                  Ulrike Prange},
  editor       = {Hartmut Ehrig and
                  Gregor Engels and
                  Francesco Parisi{-}Presicce and
                  Grzegorz Rozenberg},
  title        = {Adhesive High-Level Replacement Categories and Systems},
  booktitle    = {Graph Transformations, Second International Conference, {ICGT} 2004,
                  Rome, Italy, September 28 -- October 2, 2004, Proceedings},
  series       = {Lecture Notes in Computer Science},
  volume       = {3256},
  pages        = {144--160},
  publisher    = {Springer},
  year         = {2004},
  doi          = {10.1007/978-3-540-30203-2\_12},
}

@article{HabelP09,
  author       = {Annegret Habel and
                  Karl{-}Heinz Pennemann},
  title        = {Correctness of high-level transformation systems relative to nested
                  conditions},
  journal      = {Math. Struct. Comput. Sci.},
  volume       = {19},
  number       = {2},
  pages        = {245--296},
  year         = {2009},
  doi          = {10.1017/S0960129508007202}
}

@book{HT20,
  author       = {Reiko Heckel and
                  Gabriele Taentzer},
  title        = {Graph Transformation for Software Engineers -- With Applications to
                  Model-Based Development and Domain-Specific Language Engineering},
  publisher    = {Springer},
  year         = {2020},
  doi          = {10.1007/978-3-030-43916-3}
}

@article{JohnKLT23,
  author       = {Stefan John and
                  Jens Kosiol and
                  Leen Lambers and
                  Gabriele Taentzer},
  title        = {A graph-based framework for model-driven optimization facilitating
                  impact analysis of mutation operator properties},
  journal      = {Softw. Syst. Model.},
  volume       = {22},
  number       = {4},
  pages        = {1281--1318},
  year         = {2023},
  doi          = {10.1007/S10270-022-01078-X}
}

@article{LS05,
  author       = {Stephen Lack and
                  Pawel Sobocinski},
  title        = {Adhesive and quasiadhesive categories},
  journal      = {{RAIRO} Theor. Informatics Appl.},
  volume       = {39},
  number       = {3},
  pages        = {511--545},
  year         = {2005},
  doi          = {10.1051/ITA:2005028}
}

@article{LBKST19,
  author       = {Leen Lambers and
                  Kristopher Born and
                  Jens Kosiol and
                  Daniel Strüber and
                  Gabriele Taentzer},
  title        = {Granularity of conflicts and dependencies in graph transformation
                  systems: {A} two-dimensional approach},
  journal      = {J. Log. Algebraic Methods Program.},
  volume       = {103},
  pages        = {105--129},
  year         = {2019},
  doi          = {10.1016/J.JLAMP.2018.11.004}
}

@article{LO21,
  author       = {Leen Lambers and
                  Fernando Orejas},
  title        = {Transformation rules with nested application conditions: Critical
                  pairs, initial conflicts {\&} minimality},
  journal      = {Theor. Comput. Sci.},
  volume       = {884},
  pages        = {44--67},
  year         = {2021},
  doi          = {10.1016/J.TCS.2021.07.023}
}

@inproceedings{LBOST18,
  author       = {Leen Lambers and
                  Kristopher Born and
                  Fernando Orejas and
                  Daniel Strüber and
                  Gabriele Taentzer},
  editor       = {Reiko Heckel and
                  Gabriele Taentzer},
  title        = {Initial Conflicts and Dependencies: Critical Pairs Revisited},
  booktitle    = {Graph Transformation, Specifications, and Nets -- In Memory of Hartmut
                  Ehrig},
  series       = {Lecture Notes in Computer Science},
  volume       = {10800},
  pages        = {105--123},
  publisher    = {Springer},
  year         = {2018},
  doi          = {10.1007/978-3-319-75396-6\_6}
}

@inproceedings{LKST19,
  author       = {Leen Lambers and
                  Jens Kosiol and
                  Daniel Strüber and
                  Gabriele Taentzer},
  editor       = {Esther Guerra and
                  Fernando Orejas},
  title        = {Exploring Conflict Reasons for Graph Transformation Systems},
  booktitle    = {Graph Transformation -- 12th International Conference, {ICGT} 2019,
                  Held as Part of {STAF} 2019, Eindhoven, The Netherlands, July 15--16,
                  2019, Proceedings},
  series       = {Lecture Notes in Computer Science},
  volume       = {11629},
  pages        = {75--92},
  publisher    = {Springer},
  year         = {2019},
  doi          = {10.1007/978-3-030-23611-3\_5}
}

@inproceedings{LSTBH18,
  author       = {Leen Lambers and
                  Daniel Strüber and
                  Gabriele Taentzer and
                  Kristopher Born and
                  Jevgenij Huebert},
  editor       = {Michel Chaudron and
                  Ivica Crnkovic and
                  Marsha Chechik and
                  Mark Harman},
  title        = {Multi-granular conflict and dependency analysis in software engineering
                  based on graph transformation},
  booktitle    = {Proceedings of the 40th International Conference on Software Engineering,
                  {ICSE} 2018, Gothenburg, Sweden, May 27 -- June 03, 2018},
  pages        = {716--727},
  publisher    = {{ACM}},
  year         = {2018},
  doi          = {10.1145/3180155.3180258}
}

@article{LauerKT24,
  author       = {Alexander Lauer and
                  Jens Kosiol and
                  Gabriele Taentzer},
  title        = {Empowering model repair: a rule-based approach to graph repair without
                  side effects - extended version},
  journal      = {Innov. Syst. Softw. Eng.},
  volume       = {20},
  number       = {4},
  pages        = {597--618},
  year         = {2024},
  doi          = {10.1007/S11334-024-00587-W}
}

@inproceedings{Rensink04,
  author       = {Arend Rensink},
  editor       = {Hartmut Ehrig and
                  Gregor Engels and
                  Francesco Parisi{-}Presicce and
                  Grzegorz Rozenberg},
  title        = {Representing First-Order Logic Using Graphs},
  booktitle    = {Graph Transformations, Second International Conference, {ICGT} 2004,
                  Rome, Italy, September 28 -- October 2, 2004, Proceedings},
  series       = {Lecture Notes in Computer Science},
  volume       = {3256},
  pages        = {319--335},
  publisher    = {Springer},
  year         = {2004},
  doi          = {10.1007/978-3-540-30203-2\_23}
}

@inproceedings{LauerKT25,
	author       = {Alexander Lauer and
	Jens Kosiol and
	Gabriele Taentzer},
	editor       = {J{\"{o}}rg Endrullis and
	Matthias Tichy},
	title        = {Granular Conflict Analysis for Transformation Rules with Application
	Conditions},
	booktitle    = {Graph Transformation -- 18th International Conference, {ICGT} 2025,
	Held as Part of {STAF} 2025, Koblenz, Germany, June 11--12, 2025, Proceedings},
	series       = {Lecture Notes in Computer Science},
	volume       = {15720},
	pages        = {63--90},
	publisher    = {Springer},
	year         = {2025},
	doi          = {10.1007/978-3-031-94706-3\_4}
}

@book{fowler2018refactoring,
  title={Refactoring: improving the design of existing code},
  author={Fowler, Martin},
  year={2018},
  publisher={Addison-Wesley Professional}
}

@inproceedings{RensinkC25,
  author       = {Arend Rensink and
                  Andrea Corradini},
  editor       = {Nils Jansen and
                  Sebastian Junges and
                  Benjamin Lucien Kaminski and
                  Christoph Matheja and
                  Thomas Noll and
                  Tim Quatmann and
                  Mari{\"{e}}lle Stoelinga and
                  Matthias Volk},
  title        = {On Categories of Nested Conditions},
  booktitle    = {Principles of Verification: Cycling the Probabilistic Landscape --
                  Essays Dedicated to Joost-Pieter Katoen on the Occasion of His 60th
                  Birthday, Part {I}},
  series       = {Lecture Notes in Computer Science},
  volume       = {15260},
  pages        = {393--418},
  publisher    = {Springer},
  year         = {2024},
  doi          = {10.1007/978-3-031-75783-9\_16}
}

@article{FritscheLKST26,
  author       = {Lars Fritsche and
                  Alexander Lauer and
                  Maximilian Kratz and
                  Andy Schürr and
                  Gabriele Taentzer},
  title        = {Using weakest application conditions to rank graph transformations
                  for graph repair},
  journal      = {Log. Methods Comput. Sci.},
  volume       = {22},
  number       = {1},
  year         = {2026},
  doi          = {10.46298/LMCS-22(1:10)2026}
}

@inproceedings{ArendtBJKT10,
  author       = {Thorsten Arendt and
                  Enrico Biermann and
                  Stefan Jurack and
                  Christian Krause and
                  Gabriele Taentzer},
  editor       = {Dorina C. Petriu and
                  Nicolas Rouquette and
                  {\O}ystein Haugen},
  title        = {{Henshin: Advanced Concepts and Tools for In-Place {EMF} Model Transformations}},
  booktitle    = {Model Driven Engineering Languages and Systems -- 13th International
                  Conference, {MODELS} 2010, Oslo, Norway, October 3--8, 2010, Proceedings,
                  Part {I}},
  series       = {Lecture Notes in Computer Science},
  pages        = {121--135},
  publisher    = {Springer},
  year         = {2010},
  doi          = {10.1007/978-3-642-16145-2\_9}
}

@phdthesis{Pennemann2009,
  author       = {Karl{-}Heinz Pennemann},
  title        = {Development of correct graph transformation systems},
  school       = {University of Oldenburg, Germany},
  year         = {2009},
  url          = {https://nbn-resolving.org/urn:nbn:de:gbv:715-oops-9483},
  urn          = {urn:nbn:de:gbv:715-oops-9483}
}
	\newpage
	\appendix
	\section{Additional formal Results and Proofs}
\label{app:proofs}
First, we show a property of pullbacks which we will use throughout the proofs.
\begin{lemma}\label{lem:included_pullback}
	Consider the diagram below so that $g'' = m \circ g$ and $f'' = m \circ f$, the inner square $(g \circ f', f \circ g')$ is a pullback if and only if the outer square ($g''\circ f', f'' \circ g')$ is a pullback.
	
	\begin{center}
		\begin{tikzpicture}
			\matrix (m) [	matrix of math nodes,
			nodes in empty cells,
			row sep=1em,
			column sep=1em,
			minimum width=1em]
			{
				& &  & & B & & & & \\ \\
				X & & A & & & & D & & E\\ \\
				& & &   & C \\};
			\path[-stealth]
			(m-3-1) edge[] node[fill=white] {\edge{x}}(m-3-3)
			edge[bend left=20] node[fill=white] {\edge{h}}(m-1-5)
			edge[bend right = 20] node[fill=white] {\edge{k}}(m-5-5)
			(m-3-3) edge[] node[fill=white] {\edge{f'}}(m-1-5)
			edge[] node[fill=white] {\edge{g'}}(m-5-5)
			
			(m-5-5) edge[] node[fill=white] {\edge{f}}(m-3-7)
			(m-1-5) edge[] node[fill=white] {\edge{g}}(m-3-7)
			
			(m-3-7) edge[\rmono] node[fill=white] {\edge{m}}(m-3-9)
			
			(m-1-5) edge[] node[fill=white] {\edge{g''}}(m-3-9)
			(m-5-5) edge[] node[fill=white] {\edge{f''}}(m-3-9)
			
			;
			
		\end{tikzpicture}
	\end{center}
\end{lemma}
\begin{proof}
	\begin{enumerate}
		\item Let us first assume that the inner square $(g \circ f', f \circ g')$ is a pullback. 
		We need to show that the outer square is also a pullback. 
		Let $h \colon X \to B$ and $k \colon X \to C$ be morphisms so that $f'' \circ k = g'' \circ h$. 
		For the outer square to be a pullback, we need to show that there is a morphism $x \colon X \to A$ with $h = f' \circ x$ and $k = g' \circ x$. In particular, we have 
		\begin{align*}
			g'' \circ h &= f'' \circ k &&\iff \\
			m \circ f \circ k &= m \circ g \circ h &&\iff \\
			f \circ k &= g \circ h.
		\end{align*}
		Now, because the inner square is a pullback, there must be a morphism $x \colon X \to A$ with $k = g' \circ x$ and $h = f' \circ x$.
		\item Let us assume that the outer square $(f'' \circ g', g'' \circ f')$ is a pullback. 
		Now we want to show that the inner square is a pullback. 
		Again, let $h \colon X \to B$ and $k \colon X \to C$ be morphisms so that $g \circ h = f \circ k$.
		This implies that 
		\begin{align*}
			g \circ h &= f \circ k &&\iff \\
			m \circ g \circ h &= m \circ  f \circ k &&\iff \\
			g''\circ h &= f'' \circ k.
		\end{align*}
		Because the outer square is pullback, there is a morphism $x \colon X \to A$ with $h = f' \circ x$ and $k = g' \circ x$. \qed
	\end{enumerate}
\end{proof}

\subsection{Transfer of Definitions and Results from~\cite{LauerKT25}}\label{sec:previous_results}

In the following we present the definition of the shift operator for proto-essence, the definition of disabling essences and the results of~\cite{LauerKT25} that directly apply to our extended construction of proto-essences presented in~\cref{chapter:essences}.

\begin{definition}[Shift of proto-essence by insertion {\cite[Definition 8]{LauerKT25}}]\label{def:shift_proto_essence}
	Given two rules $\rleAC{1}$ and $\rleAC{2}$ and a proto essence by insertion $pe \coloneqq (a \circ c \colon C \to R_1, b_i \circ c \colon C \to P_i) \in \IPEss{}{\rho_1}{\rho_2}$, the \emph{shift of $pe$ over $\rho_1$}, denoted by $\PEShift{\rho_1}{pe}$, is defined as 
	\begin{equation*}
		\PEShift{\rho_1}{pe} = (a' \circ c', b_i' \circ c')
	\end{equation*}
	where (1) in the diagram below is a pullback, (IPO) is the initial pushout over $k_i$ (both stem from the computation of $pe$ and (2) is a pushout constructed by computing the pushout complement of $b_i$ and $k_i$. 
	
	\begin{center}
		\begin{tikzpicture}
			\matrix (m) [	matrix of math nodes,
			nodes in empty cells,
			row sep=1em,
			column sep=1em,
			minimum width=1em]
			{
				&    & & B_{}& &  C_{}& &  \\
				&    & & & \text{\scriptsize(IPO)} & & & \\
				&    & & A_i'& &   A_i& & & & A_i'\\
				&    & & &\text{\scriptsize (1)} &  & & & \text{\scriptsize (2)} \\
				&L_1 & & K_1 & & R_1 & & P_i & & P_i'\\};
			
			\path[-stealth]
			(m-1-4) edge[] node[fill=white] {\edge{c'}} (m-3-4)
			edge[\rmono] node[fill=white] {\edge{k_i'}} (m-1-6)
			(m-3-4) edge[] node[] {}(m-5-4)
			edge[] node[fill=white] {\edge{a'}}(m-5-2)
			edge[\rmono] node[fill=white] {\edge{k_i}}(m-3-6)
			
			(m-5-4) edge[\rmono] node[fill=white] {\edge{r_1}}(m-5-6)
			edge[\lmono] node[fill=white] {\edge{l_1}}(m-5-2)
			
			(m-5-4) edge[\rmono] node[fill=white] {\edge{r_1}}(m-5-6)
			(m-3-6) edge node[fill=white]{\edge{a}} (m-5-6)
			edge node[fill=white]{\edge{b_i}} (m-5-8)
			(m-1-6) edge node[fill=white]{\edge{c}} (m-3-6)
			(m-3-10) edge[\lmono] node[fill=white]{\edge{k_i}}(m-3-6)
			edge[] node[fill=white]{\edge{b_i'}}(m-5-10)
			(m-5-10) edge[\lmono] node[fill=white] {\edge{p'_{P'_i}}}(m-5-8)
			;
		\end{tikzpicture}
	\end{center}
\end{definition}

\begin{definition}[Disabling essence {\cite[Definition 9]{LauerKT25}}]
	\label{def:asynchronous_essences}
	Given two rules $\rleAC{1}$ and $\rleAC{2}$, a proto essence by deletion $(a \circ c \colon C \to L_1, b_i \circ c \colon C \to P_i) \in \DPEss{}{\rho_1}{\rho_2}$ is a \emph{disabling essence by deletion} of $\rho_1$ in $\rho_2$ if 
	\begin{enumerate}
		\item if it is embedded in a transformation pair $(\trans{L_1P_i}{e_{L_1}}{\rho_1}{H_1}, \trans{L_1P_i}{e_{P_i} \circ p_{L_2}}{\rho_2}{H_2})$ where $(e_{L_1} \colon L_1 \to L_1P_i, e_{P_i} \colon P_i \to L_1P_i) \in \Epi$,
		\item $P_i$ is existentially bound in $ac_2$ or $P_i = L_2$, and 
		\item $(a,b_i) \neq (!_{L_1}, !_{P_i})$,
	\end{enumerate} 
	Given a proto-essence by insertion $(a \circ c \colon C \to R_1, b_i \circ c \colon C \to P_i) \in \IPEss{}{\rho_1}{\rho_2}$, the span $\PEShift{\rho_1}{a \circ c, b_i \circ c}$ is a \emph{disabling essence by insertion} of $\rho_1$ in $\rho_2$ if 
	\begin{enumerate}
		\item if it is embedded in a transformation pair $(\trans{L_1P_i}{e_{L_1}}{\rho_1}{H_1}, \trans{L_1P_i}{e_{P_i} \circ p_{L_2}}{\rho_2}{H_2})$ where $(e_{L_1} \colon L_1 \to L_1P_i, e_{P_i} \colon P_i \to L_1P_i) \in \Epi$,
		\item $P_i$ is universally bound in $ac_2$, and 
		\item $(a_i, b_i) \neq (!_{R_1}, !_{P_i})$. 
	\end{enumerate}
	The set of all disabling essences of two rules $\rho_1$ and $\rho_2$ is denoted by $\ESS{\rho_1}{\rho_2}$.

\end{definition}

The following proposition shows that the set of disabling essences for rules with application conditions is an extension of disabling essences for plain rules~\cite{ACR19}.
For this, we show that the set of disabling essences for rules with application conditions contains the set disabling essences for plain rules. 
Moreover, each disabling essence, where both morphisms have the \acp{LHS} of the rules as codomain, are also disabling essences for plain rules. 
This means that our approach exactly computes disabling essences for plain rules.
\begin{proposition}[Containment of disabling essence for plain rules]\label{prop:containment_proto_essences}
	Given two rules $\rho_1 = \rleAC{1}$, $\rho_2 = \rleAC{2}$, the set of disabling essence by deletion of $\rho_1$ in $\rho_2$ contains the set of disabling essences for the plain rules of $\rlePlain{1}$ in $\rlePlain{2}$. 
	Moreover, each disabling essence by deletion of the form $ b_1 \circ c \colon C \to L_1, b_2 \circ c \colon C \to L_2$ is a disabling essence of $\rlePlain{1}$ in $\rlePlain{2}$.
\end{proposition}

\begin{proof}[of \cref{prop:containment_proto_essences}]
	Let a disabling essence by deletion $de = (b_1 \circ c \colon C \to L_1, b_2 \circ c \colon C \to L_2) \in \DEss{\rho_1}{\rho_2}$ be given. 
	By \cref{def:essencesAC} there is either a pair of morphism $(e_{L_1} \colon L_1 \to L_1P_n, e_{P_n} \colon P_n \to L_1P_n) \in \Epi$, where $P_n$ is either a leaf of the tree-repesentation of $ac_2$ or the \ac{LHS} $L_2$ of $\rho_2$  so that $\DPEss{n}{e_{L_1}}{e_{P_n}} = de$. 
	In both cases, the morphism $b_2 \circ c$ has codomain $L_2$ if and only if
	$\DPEss{n}{e_{L_1}}{e_{P_n}} = \DPEss{0}{e_{L_1}}{e_{P_n}} = \ess_{l_1}(e_{L_1}, e_{P_n} \circ \ldots \circ p_1) \neq (!_{L_1}, !_{L_2})$. 
	Because the first iteration of our construction is identical to the construction of disabling essences for plain rules \cite[Definition 3.8]{ACR19}, $de$ is also a disabling essence for plain rules. 
	Finally, we must show that each disabling essence for plain rules is contained in $\DEss{\rho_1}{\rho_2}$:
	Disabling essence of plain rules are defined as $\ess_{l_1}(m_1,m_2)$ for some cospan $(m_1 \colon L_1 \to G, m_2 \colon L_2 \to G)$. 
	I.e. the construction starts by calculating the pullback of $(m_1,m_2)$ as shown in the diagram below. By \EMfactorization\ there is a pair of morphisms $(e_{L_1} \colon L_1 \to L_1L_2, e_{L_2} \colon L_2 \to L_1L_2) \in \Epi$ and a monomorphism $m^* \colon L_1L_2 \mono G \in \Mono$ with $m_1 = m^* \circ e_{L_1}$ and $m_2 = m^* \circ e_{L_2}$. By \cref{lem:included_pullback} the square (1) is also a pullback and $\ess_{l_1}(m_1,m_2) = \ess_{l_1}(e_{L_1}, e_{L_2}) = \DPEss{n}{e_{L_1}}{e_{L_2}}$, i.e., this is also a disabling essence for rules with application conditions. \qed
	
	\begin{center}
		\begin{tikzpicture}
			\matrix (m) [	matrix of math nodes,
			nodes in empty cells,
			row sep=1em,
			column sep=1em,
			minimum width=1em]
			{
				& & L_1 & & & & \\ 
				A & & \text{\scriptsize(1)}& & L_1L_2 & & G\\ 
				&   & L_2 \\};
			\path[-stealth]
			\Edge{2-1}{}{fill=white}{b_1}{1-3}
			\Cedge{}{fill=white}{b_2}{3-3}			
			
			\Edge{3-3}{bend right = 10}{fill=white}{m_2}{2-7}
			\Cedge{}{fill=white}{e_{L_2}}{2-5}
			
			\Edge{1-3}{bend left = 10}{fill=white}{m_1}{2-7}
			\Cedge{}{fill=white}{e_{L_1}}{2-5}
			
			\Edge{2-5}{\rmono}{fill=white}{m^{*}}{2-7}
			
			;
			
		\end{tikzpicture}
	\end{center}
\end{proof}

Finally, we show that a disabling essence is embedded in each parallel dependent transformation pair.
\begin{theorem}\label{thm:Embedding_of_disabling_essence}
	For each pair of parallel dependent transformations $(t_1,t_2)$ via rules $\rho_1$ and $\rho_2$, there is a disabling essence $de \in \ESS{\rho_1}{\rho_2}$ that is embedded in $(t_1,t_2)$.
\end{theorem}

\begin{proof}[of \Cref{thm:Embedding_of_disabling_essence}]
	Given a pair of transformations $(t_1,t_2) = (\trans{G}{\rho_1}{m_1}{H_1}, \trans{G}{\rho_2}{m_2}{H_2})$ via rules $\rho_1 = \rleAC{1}$ and $\rho_2 = \rleAC{2}$ so that $t_1$ causes a conflict for $t_2$.
	For this, there are two possibilities, either $t_1$ destroys the match $m_2$ of $t_2$, or it introduces a violation of the application condition $ac_2$ of $\rho_2$.
	If the conflict is caused by destroying the match $m_2$, \cref{prop:containment_proto_essences} implies that a disabling essence of $\rho_1$ in $\rho_2$ embeds into the transformation pair \cite[Theorem 3.13]{ACR19}. 
	
	Otherwise, if the application condition $ac_2$ is violated, the proof is analogous to the proof of \cite[Theorem 2]{LauerKT25}. 
	The only difference is that we consider morphisms that lead to the leaves of the tree representation of the application condition (instead of just considering the morphisms in a condition in \ac{ANF}) and that we need to obtain a morphism pair contained in $\Epi$ to argue that the embedded span is actually contained in $\DEss{\rho_1}{\rho_2}$. This morphism pair can obtained by using \EMfactorization\ at an appropriate point and using \Cref{lem:included_pullback}. \qed
	
\end{proof}

\subsection{Proofs}

\begin{proof}[of \cref{thm:disjointness_ac-disregarding}]
	Let two rules $\rho_i = \rleAC{i}$, with $i = 1,2$, an ac-conflicting rule overlap $ro = (b_i \colon A \to P_i, b_j \colon A \to P_j, p_{L_1}, p_{L_2})$ of $\rho_1$ and $\rho_2$, and a transformation pair $(t_1,t_2)$ be given so that $ro$ is embedded in $(t_1,t_2)$. 
	We want to show that $(t_1,t_2)$ is parallel independent ac-disregaring.
	Because $ro$ is ac-conflicting, there is an ac-disregarding parallel independent transformation pair $(t_1',t_2') = (\trans{G'}{\rho_1}{m_1'}{H_1'}, \trans{G'}{\rho_2}{m_2'}{H_2'})$ so that $ro$ is embedded via morphisms $m_1^i$ and $m_2^j$ as shown in the diagram below. 
	We can construct the pullback of $m_1'$ and $m_2'$ by calculating the pullbacks (2), (3) and (4). 
	Because $t_1'$ and $t_2'$ are parallel independent ac-disregarding, the morphisms $l_1'$ and $l_2'$ constructed by calculating the pullbacks (5) and (6) are isomophisms. 
	We can use (2), (3) and (4) to also construct the pullback of $m_1$ and $m_2$ which results in the same morphisms as for $m_1'$ and $m_2'$.
	This implies that $t_1$ and $t_2$ are parallel independent ac-disregarding.
	
	\begin{center}
		\begin{tikzpicture}
			\matrix (m) [	matrix of math nodes,
			nodes in empty cells,
			row sep=1em,
			column sep=1em,
			minimum width=1em]
			{
				K_1' & & & & A_{}' & & & & K_2' \\ \\
				& \text{\scriptsize(5)}& & D_1 & \text{\scriptsize(4)} & D_2 & & \text{\scriptsize(6)} \\ \\
				K_1& & L_1&  \text{\scriptsize(2)} & A_{} & \text{\scriptsize(3)} & L_2 && K_2 \\ \\
				& &  & P_i& \text{\scriptsize(1)}& P_j &  & &  \\ \\
				& & & &  G_{}/G'\\
			};
			\path[-stealth] 
			
			\Edge{1-1}{}{}{}{5-1}
			\Cedge{\rmono}{fill=white}{l_1'}{1-5}
			
			\Edge{1-9}{}{}{}{5-9}
			\Cedge{\lmono}{fill=white}{l_2'}{1-5}
			
			\Edge{1-5}{}{}{}{3-4}
			\Cedge{}{}{}{3-6}
			
			\Edge{3-4}{}{}{}{5-3}
			\Cedge{}{}{}{5-5}
			
			\Edge{3-6}{}{}{}{5-7}
			\Cedge{}{}{}{5-5}
			
			\Edge{5-5}{}{}{}{7-4}
			\Cedge{}{}{}{7-6}
			
			\Edge{5-3}{}{fill=white}{p_{L_1}}{7-4}
			\Cedge{bend right = 40}{fill= white}{m_1/m_1'}{9-5}
			\Edge{5-7}{}{fill=white}{p_{L_2}}{7-6}
			\Cedge{bend left=40}{fill= white}{m_1/m_1'}{9-5}
			
			\Edge{7-4}{}{fill=white}{m_1^i/m_1^{i'}}{9-5}
			\Edge{7-6}{}{fill=white}{m_2^j/m_2^{j'}}{9-5}
			
			\Edge{5-1}{\rmono}{fill=white}{l_1}{5-3}
			\Edge{5-9}{\lmono}{fill=white}{l_2}{5-7}

			;
		\end{tikzpicture}
	\end{center}
	
	We show the second part of the claim by contradiction: If a non-ac-conflicting rule overlap is embedded in a ac-disregarding parallel independent transformation, this implies that the rule overlap is ac-conflicting (by definition). \qed
\end{proof}

\begin{proof}[of \cref{thm:embedding_of_ac_conflicting_essences}]
	Given a pair of transformations $(t_1,t_2) = (\trans{G}{\rho_1}{m_1}{H_1}, \trans{G}{\rho_1}{m_1}{H_1})$ via rules $\rho_1 = \rleAC{1}$ and $\rho_2 = \rleAC{2}$ and a disabling essence $de = (b_1 \colon A \to L_1, b_j \colon A \to P_j, id_{L_1}, p_{L_2})$ that is embedded in $(t_1,t_2)$ via morphisms $m_1 \colon L_1 \to G$ and $m_2^j \colon P_j \to G$ as shown in the diagram below.
	To show the claim, we must show that (a) $l_2'$ (constructed by the pullbacks (2) and (3) shown in the diagram below) is an isomorphism if $de$ is ac-conflicting and that (b) $de$ is ac-conflicting if $l_2'$ is an isomorphism.
	We start by showing (a): Because $de$ is ac-conflicting, we know that the transformation pair $(t_1,t_2)$ is parallel independent ac-disregarding (\cref{def:ac-conflicting} and \cref{thm:disjointness_ac-disregarding}). 
	The square (1) + (2) (as composition of pullbacks) is a pullback of $m_1$ and $m_2$. 
	By the definition of parallel independence (\cref{def:parallel_independence}), the morphism $l_2'$ obtained by calculating the pullback (3) (and also by calculating (2) and (3)) must be an isomorphism. 
	
	To show (b), we assume that $l_2'$ is an isomorphism and we must show that the embedding of $de$ implies that $t_1$ and $t_2$ are parallel independent ac-disregarding. 
	Let us first show that $t_1$ does not cause a conflict for $t_2$: 
	\begin{center}
		\begin{tikzpicture}
			\matrix (m) [	matrix of math nodes,
			nodes in empty cells,
			row sep=1em,
			column sep=1em,
			minimum width=1em]
			{
				& &  & & & A_{} & & & L_2' & & K_2' \\ 
				& & & & & & & \text{\scriptsize(2)} & & \text{\scriptsize(3)}\\
				R_1 & & K_1 & & L_1 & \text{\scriptsize(1)}& P_j & & L_2 & & K_2 & & R_2 \\ \\
				H_1 & & D_1 & & & G & & & & & D_2 & & H_2 \\
			};
			\path[-stealth] 
			\Edge{1-6}{}{fill=white}{b_1}{3-5}
			\Cedge{}{fill=white}{b_j}{3-7}
			
			\Edge{1-9}{}{fill=white}{p'_{L_2}}{1-6}
			\Cedge{}{}{}{3-9}
			
			\Edge{1-11}{\lmono}{fill=white}{l_2'}{1-9}
			\Cedge{}{}{}{3-11}
			
			\Edge{3-11}{\lmono}{fill=white}{l_2}{3-9}
			\Cedge{}{}{}{5-11}
			\Cedge{\rmono}{fill=white}{r_2}{3-13}
			
			\Edge{5-11}{\lmono}{fill=white}{g_2}{5-6}
			\Cedge{\rmono}{fill=white}{h_2}{5-13}
			
			\Edge{3-13}{}{}{}{5-13}
			\Edge{3-1}{}{}{}{5-1}

			\Edge{5-3}{\rmono}{fill=white}{g_1}{5-6}
			\Cedge{\lmono}{fill=white}{h_1}{5-1}
			
			\Edge{3-3}{\rmono}{fill=white}{l_1}{3-5}
			\Cedge{\lmono}{fill=white}{r_1}{3-1}
			\Cedge{}{}{}{5-3}
			
			\Edge{3-5}{}{fill=white}{m_1}{5-6}
			\Edge{3-9}{}{fill=white}{m_2}{5-6}
			\Cedge{}{fill=white}{p_{L_2}}{3-7}
			
			\Edge{3-7}{}{fill=white}{m_2^j}{5-6}
			;
		\end{tikzpicture}
	\end{center}
	Again, the square (1) + (2) is a pullback of of $m_1$ and $m_2$. 
	Since $l_2'$ is an isomorphism, \cref{def:parallel_independence} implies that $t_1$ does not cause a conflict for $t_2$.

	To show that $t_2$ does not cause a conflict for $t_1$ we must consider whether $de$ is an inserting or a deleting essence. 
	We start by assuming that $de$ is a deleting essence:
		
	\begin{center}
		\begin{tikzpicture}
			\matrix (m) [	matrix of math nodes,
			nodes in empty cells,
			row sep=1em,
			column sep=1em,
			minimum width=1em]
			{
				K_1' & & & L_1' \\ 
				& \text{\scriptsize(4)}\\
				K_1 & & L_1 & \text{\scriptsize(1) + (2)}& L_2 \\ \\
				& & & G\\ 
			};
			\path[-stealth] 
			\Edge{1-1}{\rmono}{fill=white}{l_1'}{1-4}
			\Cedge{}{}{}{3-1}
			\Edge{3-1}{\rmono}{fill=white}{l_1}{3-3}
			
			\Edge{3-3}{}{fill=white}{m_1}{5-4}

			\Edge{3-5}{}{fill=white}{m_2}{5-4}

			\Edge{1-4}{}{}{}{3-3}
			\Cedge{}{}{}{3-5}
			
			;
		\end{tikzpicture}
	\end{center}
	Since $de$ is a disabling essence of $\rho_1$ in $\rho_2$, we know that $(b_1, b_j)$ is calculated as pullback of some morphism pair $(e_{L_1}, e_{P_n} \circ \ldots \circ p_{j+1}) \in \Epi$ during the construction of disabling essences. Moreover, we know that (1) + (2) and (4) are also computed during the computation of the disabling essences (in the first iteration in \cref{def:essencesAC}).
	This implies that $l_1'$ must be an isomorphism (as otherwise $de$ would not be a disabling essence of $\rho_1$ and $\rho_2$ when starting the construction with $(e_{L_1}, e_{P_n} \circ \ldots \circ p_{j+1})$). 
	I.e., $t_2$ does not cause a conflict for $t_1$.
	
	\begin{center}
		\begin{tikzpicture}
			\matrix (m) [	matrix of math nodes,
			nodes in empty cells,
			row sep=1em,
			column sep=1em,
			minimum width=1em]
			{
				D & & L_2' \\ 
				& \text{\scriptsize(6)}\\
				A & & A \\ 
				& \text{\scriptsize(5)}\\
				K_{} && L_1 && P_j \\ \\
				D_1 & & & G \\
			};
			
			\path[-stealth] 
			\Edge{1-1}{}{}{}{3-1}
			\Cedge{\rmono}{fill=white}{$k$}{1-3}
			
			\Edge{3-1}{\rmono}{fill=white}{id_{A}}{3-3}
			\Cedge{}{fill=white}{a}{5-1}
			
			\Edge{5-1}{\rmono}{fill=white}{l_1}{5-3}
			\Cedge{}{}{}{7-1}
			\Edge{7-1}{\rmono}{fill=white}{g_1}{7-4}
			
			\Edge{1-3}{}{fill=white}{p'_{L_2}}{3-3}
			\Edge{3-3}{}{fill=white}{b_1}{5-3}
			\Cedge{}{fill=white}{b_j}{5-5}
			
			\Edge{5-3}{}{fill=white}{m_1}{7-4}
			\Edge{5-5}{}{fill=white}{m_2^j}{7-4}
			;
		\end{tikzpicture}
	\end{center}
	
	If $ce$ is a inserting essence, by construction of the shift of proto-essences (\cref{def:shift_proto_essence}), there is a morphism $a \colon A \to L_1$ with $b_1 = l_1 \circ a$. 
	This implies that the square (5) shown in the figure above is a pullback. 
	Then, we construct pullback (6) (where $L_2'$ is obtained by computing the pullback of $m_1$ and $m_2$ as shown above) and the square (5) + (6) is also a pullback (as composition of pullbacks). 
	Because pullbacks are stable under isomorphisms, $k \colon D \to L_2'$ is an isomorphism, i.e., $t_2$ does not cause a conflict for $t_1$. \qed

\end{proof}

	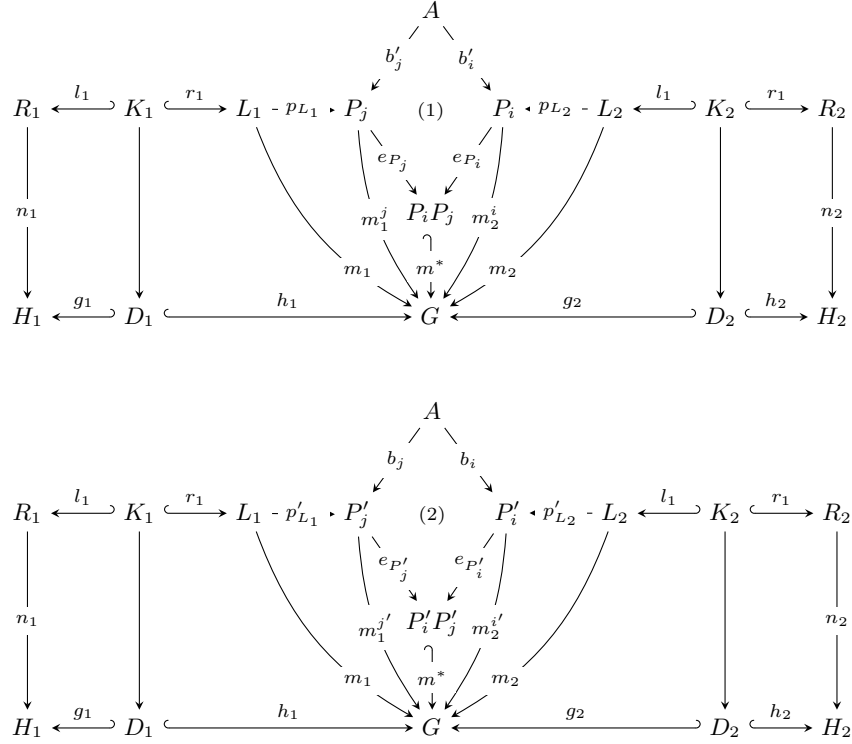
\begin{figure}[t]
	\begin{tikzpicture}
		\matrix (m) [	matrix of math nodes,
		nodes in empty cells,
		row sep=1em,
		column sep=.8em,
		minimum width=1em]
		{	
			& & & &   & & &   & & &  \\
			& & & &   & & &  A & & &  \\				\\			
			R_1 & & K_1 & & L_1 & & P_j & \text{\scriptsize(1)} & P_i & & L_2 & & K_2 & & R_2 \\ \\
			& & & & & & &  P_iP_j \\				\\
			H_1 & & D_1 & &     & &  & G_{}&    & &     & & D_2 & & H_2 \\
		};
		\path[-stealth]

		(m-4-3) edge[\lmono] node [above] {\edge{l_1}}(m-4-1)
		edge[\rmono] node [above] {\edge{r_1}}(m-4-5)
		
		(m-8-3) edge[\lmono] node [above] {\edge{g_1}}(m-8-1)
		edge[\rmono] node [above] {\edge{h_1}}(m-8-8)
		
		(m-4-13) edge[\lmono] node [above] {\edge{l_1}}(m-4-11)
		edge[\rmono] node [above] {\edge{r_1}}(m-4-15)
		
		(m-8-13) edge[\rmono] node [above] {\edge{h_2}}(m-8-15)
		edge[\lmono] node [above] {\edge{g_2}}(m-8-8)
		
		(m-4-1) edge node [fill=white] {\edge{n_1}}(m-8-1)
		(m-4-3) edge node [] {}(m-8-3)
		
		(m-4-15) edge node [fill=white] {\edge{n_2}}(m-8-15)
		(m-4-13) edge node [] {}(m-8-13)
		
		(m-4-5) edge node [fill=white]{\edge{p_{L_1}}}(m-4-7)
		(m-4-5) edge[bend right = 20] node [fill=white, near end]{\edge{m_1}}(m-8-8)
		
		(m-4-11) edge node [fill=white]{\edge{p_{L_2}}}(m-4-9)
		(m-4-11) edge[bend left = 20] node [fill=white, near end]{\edge{m_2}}(m-8-8)
		
		(m-4-7) edge node [fill=white]{\edge{e_{P_j}}} (m-6-8)
		edge[bend right = 15] node [fill=white]{\edge{m_1^j}} (m-8-8)
		
		(m-4-9) edge node [fill=white]{\edge{e_{P_i}}} (m-6-8)
		edge[bend left = 15] node [fill=white]{\edge{m_2^i}} (m-8-8)
		
		(m-2-8) edge node [fill=white]{\edge{b_j'}}(m-4-7)
		edge node [fill=white]{\edge{b_i'}}(m-4-9)
		
		(m-6-8) edge[\lmono] node [fill=white]{\edge{m^*}}(m-8-8)

		;
		
	\end{tikzpicture}
	\begin{tikzpicture}
		\matrix (m) [	matrix of math nodes,
		nodes in empty cells,
		row sep=1em,
		column sep=.8em,
		minimum width=1em]
		{	
			& & & &   & & &   & & &  \\
			& & & &   & & &  A & & &  \\				\\			
			R_1 & & K_1 & & L_1 & & P'_j & \text{\scriptsize(2)} & P'_i & & L_2 & & K_2 & & R_2 \\ \\
			& & & & & & &  P'_iP'_j \\				\\
			H_1 & & D_1 & &     & &  & G_{}&    & &     & & D_2 & & H_2 \\
		};
		\path[-stealth]

		(m-4-3) edge[\lmono] node [above] {\edge{l_1}}(m-4-1)
		edge[\rmono] node [above] {\edge{r_1}}(m-4-5)
		
		(m-8-3) edge[\lmono] node [above] {\edge{g_1}}(m-8-1)
		edge[\rmono] node [above] {\edge{h_1}}(m-8-8)
		
		(m-4-13) edge[\lmono] node [above] {\edge{l_1}}(m-4-11)
		edge[\rmono] node [above] {\edge{r_1}}(m-4-15)
		
		(m-8-13) edge[\rmono] node [above] {\edge{h_2}}(m-8-15)
		edge[\lmono] node [above] {\edge{g_2}}(m-8-8)
		
		(m-4-1) edge node [fill=white] {\edge{n_1}}(m-8-1)
		(m-4-3) edge node [] {}(m-8-3)
		
		(m-4-15) edge node [fill=white] {\edge{n_2}}(m-8-15)
		(m-4-13) edge node [] {}(m-8-13)
		
		(m-4-5) edge node [fill=white]{\edge{p'_{L_1}}}(m-4-7)
		(m-4-5) edge[bend right = 20] node [fill=white, near end]{\edge{m_1}}(m-8-8)
		
		(m-4-11) edge node [fill=white]{\edge{p'_{L_2}}}(m-4-9)
		(m-4-11) edge[bend left = 20] node [fill=white, near end]{\edge{m_2}}(m-8-8)
		
		(m-4-7) edge node [fill=white]{\edge{e_{P'_j}}} (m-6-8)
		edge[bend right = 15] node [fill=white]{\edge{m_1^{j'}}} (m-8-8)
		
		(m-4-9) edge node [fill=white]{\edge{e_{P'_i}}} (m-6-8)
		edge[bend left = 15] node [fill=white]{\edge{m_2^{i'}}} (m-8-8)
		
		(m-2-8) edge node [fill=white]{\edge{b_j}}(m-4-7)
		edge node [fill=white]{\edge{b_i}}(m-4-9)
		
		(m-6-8) edge[\lmono] node [fill=white]{\edge{m^*}}(m-8-8)

		;
		
	\end{tikzpicture}
	\caption{Embedding of two rule overlaps in a transformation pair} \label{fig:embedding_rule_overlap}
\end{figure}

\begin{proof}[of \cref{thm:embedding_conflict_essences}]
	Given a pair of transformations $(t_1,t_2) = (\trans{G}{\rho_1}{m_1}{H_1}, \trans{G}{\rho_2}{m_2}{H_2})$ via rules $\rho_1 = \rleAC{1}$ and $\rho_2 = \rleAC{2}$, and two rule overlaps $ro = (\Span{P_j}{b_j}{A}{b_i}{P_i}, p_{L_1} \colon L_1 \to P_j, p_{L_2} \colon L_2 \to P_i)$ and $(\Span{P'_j}{b'_j}{A'}{b'_i}{P'_i}, p'_{L_1} \colon L_1 \to P'_j, p'_{L_2} \colon L_2 \to P'_i)$ that are embeddable in $(t_1,t_2)$ as shown in \Cref{fig:embedding_rule_overlap}. 
	We want to show that there is a composed overlap $co \in \Comp{ro}{ro'}$ that is also embedded in $(t_1,t_2)$. 
	To do so, well will construct a composed overlap on the basis of \cref{def:composition_overlaps} that is embedded in the transformation pair. 
	
	As shown in the diagram in \cref{fig:composed_construction} below, by \EMfactorization, there are morphism pairs $(e_{P_i} \colon P_i \to P_i^*, e_{P'_i} \colon P_i' \to P_i^*) \in \Epi$ and $(e_{P_j} \colon P_j \to P_j^*, e_{P'_j} \colon P_j' \to P_j^*)$ and morphisms $m_i^* \colon P_i^* \mono G \in \Mono$ and $m_j^* \colon P_j^* \mono G \in \Mono$ so that 
	\begin{align*}
		m_1^i &= m_i^* \circ e_{P_i}, \\
		m_2^{i'} &= m_i^* \circ e_{P'_i}, \\
		m_1^j &= m_j^* \circ e_{P_j}, \text{ and}\\
		m_2^{j'} &= m_j^* \circ e_{P'_j}. \\
	\end{align*}
	
	\begin{figure}[t]
		\centering
		\begin{tikzpicture}
				\matrix (m) [	matrix of math nodes,
			nodes in empty cells,
			row sep=1em,
			column sep=.8em,
			minimum width=1em]
			{	
				& & & & A \\ \\
				& & P_i & && &  P_j \\ \\
				P_i^* & & & & G_{ } & & & & P_j^* \\ \\
				&& P_i' && & & P_j' \\ \\
				& & & & A' \\
			};
			
			\path[-stealth] 
			
			\Edge{1-5}{}{fill=white}{b_i}{3-3}
			\Cedge{}{fill=white}{b_j}{3-7}
			
			\Edge{3-3}{}{fill=white}{m_1^i}{5-5}
			\Cedge{}{fill=white}{e_{P_i}}{5-1}
			
			\Edge{3-7}{}{fill=white}{m_2^j}{5-5}
			\Cedge{}{fill=white}{e_{P_j}}{5-9}
			
			\Edge{7-3}{}{fill=white}{m_1^{i'}}{5-5}
			\Cedge{}{fill=white}{e_{P'_i}}{5-1}
			
			\Edge{7-7}{}{fill=white}{m_2^{j'}}{5-5}
			\Cedge{}{fill=white}{e_{P'_j}}{5-9}
			
			\Edge{5-1}{\rmono}{fill=white}{m_i^*}{5-5}
			\Edge{5-9}{\lmono}{fill=white}{m_j^*}{5-5}
			
			\Edge{9-5}{}{fill=white}{b'_i}{7-3}
			\Cedge{}{fill=white}{b'_j}{7-7}
			;
		\end{tikzpicture}
		\caption{Construction of composed rule overlaps on the basis of two embeddings}\label{fig:composed_construction}
	\end{figure}
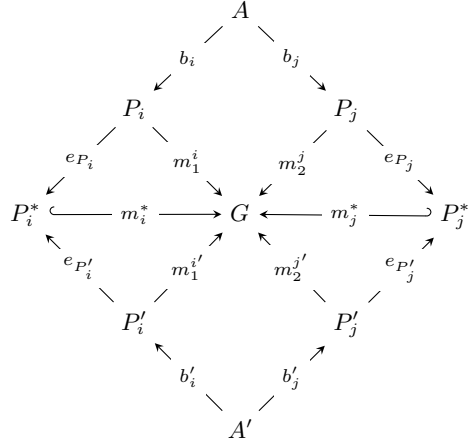
	Moreover, there is another morphism pair $(e_{P_i^*} \colon P_i^* \mono K^*, e_{P_j^*} \colon P_j^* \mono K^*) \in \Epi$ and a morphism $m^* \colon K^* \mono G \in \Mono$ so that the left side of the  diagram in \cref{fig:composition2} commutes. 
	The fact that $e_{P_i^*}$ and $e_{P_j^*}$ are $\Mono$-morphisms, follows because $m_i^*$ and $m_j^*$ are $\Mono$-morphisms and the class $\Mono$ is closed under decomposition.
	
	\begin{figure}[t]
		\centering
		\begin{tikzpicture}
			\matrix (m) [	matrix of math nodes,
			nodes in empty cells,
			row sep=1em,
			column sep=1em,
			minimum width=1em]
			{
				L_1& & A & & L_2 \\ \\ \\
				P_i & & A^*_{} & & P_j \\ \\
				& P_i^* & \text{\scriptsize (1)} & P_j^* \\ \\ 
				& & K^* \\ \\
				& & G \\
			};
			\path[-stealth] 
			
			\Edge{1-1}{}{fill=white}{p_{L_1}}{4-1}
			\Edge{1-5}{}{fill=white}{p_{L_2}}{4-5}
			
			\Edge{1-3}{}{fill=white}{b_i}{4-1}
			\Cedge{}{fill=white}{b_j}{4-5}
			
			\Edge{4-1}{}{fill=white}{e_{P_i}}{6-2}
			\Cedge{bend right = 25}{fill=white}{m_1^i}{10-3}
			
			\Edge{6-2}{\rmono}{fill=white}{e_{P^*_i}}{8-3}
			\Cedge{bend right = 10, \lmono}{fill=white}{m_i^*}{10-3}

			\Edge{6-4}{\lmono}{fill=white}{e_{P^*_j}}{8-3}
			\Cedge{bend left = 10, \rmono}{fill=white}{m_j^*}{10-3}
			
			\Edge{4-5}{}{fill=white}{e_{P_j}}{6-4}
			\Cedge{bend left = 25}{fill=white}{m_2^j}{10-3}
			
			\Edge{4-3}{\lmono}{fill=white}{b_i^*}{6-2}
			\Cedge{\rmono}{fill=white}{b_j^*}{6-4}
			
			\Edge{8-3}{\lmono}{fill=white}{m^*}{10-3}
			
			;
		\end{tikzpicture}
		\begin{tikzpicture}
			\matrix (m) [	matrix of math nodes,
			nodes in empty cells,
			row sep=1em,
			column sep=1em,
			minimum width=1em]
			{
				& & A' \\ \\ \\
				P'_i & & A^*_{} & & P'_j \\ \\
				& P_i^* & \text{\scriptsize (1)} & P_j^* \\ \\ 
				& & K^* \\ \\ 
				\\ 
			};
			\path[-stealth]

			\Edge{1-3}{}{fill=white}{b_i}{4-1}
			\Cedge{}{fill=white}{b_j}{4-5}
			
			\Edge{4-1}{}{fill=white}{e_{P'_i}}{6-2}
			\Edge{6-2}{\rmono}{fill=white}{e_{P^*_i}}{8-3}
			
			\Edge{6-4}{\lmono}{fill=white}{e_{P^*_j}}{8-3}
			
			\Edge{4-5}{}{fill=white}{e_{P'_j}}{6-4}
			
			\Edge{4-3}{\lmono}{fill=white}{b_i^*}{6-2}
			\Cedge{\rmono}{fill=white}{b_j^*}{6-4}
			
			;
		\end{tikzpicture}
		\caption{Construction of composed rule overlaps on the basis of two embeddings}
		\label{fig:composition2}
	\end{figure}
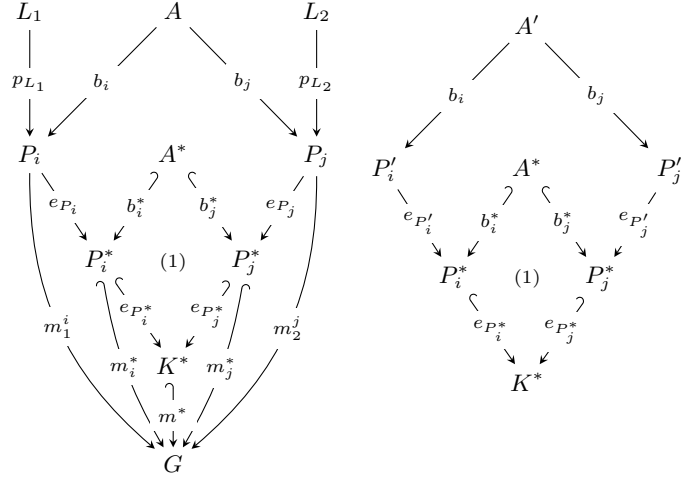
	
	We get the object $A^*$ and the morphisms $b_i^*$ and $b_j^*$ by computing the pullback $(1)$ of $(e_{P_i^*}, e_{P_j^*})$. 
	By \Cref{lem:included_pullback}, the square $(m_i^* \circ b_i^*, m_j^* \circ b_j^*)$ is also a pullback and we have
	\begin{align*}
		m_1 &= m_1^i \circ p_{L_1}  \\
		&= m_i^* \circ e_{P_i} \circ p_{L_1}.
	\end{align*}
	and
	\begin{align*}
		m_2 &= m_2^j \circ p_{L_2}  \\
		&= m_j^* \circ e_{P_j} \circ p_{L_2}.
	\end{align*} 
	This means that the rule overlap $co = (\MonoCoSpan{P_i^*}{b_i^*}{A^*}{b_j^*}{P_j^*}, e_{P_i} \circ p_{L_1}, e_{P_j} \circ p_{L_2})$ is embedded in $(t_1,t_2)$.
	Secondly, we must show that $ce \in \Comp{ro}{ro'}$:
	For this, we need to show that the squares $(e_{P_i^*} \circ e_{P_i} \circ b_j, e_{P_j^*} \circ e_{P_j} \circ b_j)$ and $(e_{P^*_i} \circ e_{P'_i} \circ b'_j, e_{P^*_j} \circ e_{P'_j} \circ b'_j)$ in the diagram in \cref{fig:composition2} are pullbacks. 
	This follows directly with \cref{lem:included_pullback} since both squares $(m_1^i \circ b_i, m_2^j \circ b_j)$ and $(m_1^{i'} \circ b_i', m_2^{j'} \circ b_j')$ are pullbacks because both rule overlaps $ro$ and $ro'$ are embedded in the transformation pair via $m_1^i$, $m_1^j$, $m_1^{i'}$ and $m_2^{j'}$, respectively. 
	
	We also need to show that if a composed essence $co \in \Comp{ro}{ro'}$ is embedded in $(t_1,t_2)$ via morphisms $m_i^* \colon P_i^* \to G$ and $m_j^* \colon P_j^* \to G$ (shown \cref{fig:composition2}), then both overlaps $ro$ and $ro'$ are embedded in $(t_1,t_2)$:
	First, we will show that $ro$ is embedded in $(t_1,t_2)$. 
	Since $co$ is a composition of $ro$ and $ro'$, there is a morphism pair $(e_{P_j^*},e_{P_i^*}) \in \Mono$ so that square (1) and the overall square (involving $b_i$ and $b_j$) is a pullback.
	The same square can be constructed by computing the pullback (2), (3), and (4) in \cref{fig:compoistion_3} which will now use to show that $ro$ is embedded in $(t_1,t_2)$:
	
	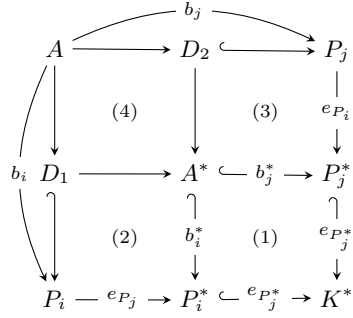
\begin{figure}
		\centering
		\begin{tikzpicture}
				\matrix (m) [	matrix of math nodes,
			nodes in empty cells,
			row sep=1em,
			column sep=1em,
			minimum width=1em]
			{
				A & & D_2  & & P_j \\ 
				 & \text{\scriptsize (4)} & & \text{\scriptsize (3)}\\
				D_1 & & A^*_{} & & P_j^* \\ 
				& \text{\scriptsize (2)} & & \text{\scriptsize (1)}\\
				P_i & & P_i^* & & K^* \\
			};
			\path[-stealth] 
			
			\Edge{1-1}{bend left = 25}{fill=white}{b_j}{1-5}
			\Cedge{bend right=25}{fill=white}{b_i}{5-1}
			\Cedge{}{}{}{1-3}
			\Cedge{}{}{}{3-1}
			
			\Edge{1-3}{}{}{}{3-3}
			\Cedge{\rmono}{}{}{1-5}
			
			\Edge{3-1}{}{}{}{3-3}
			\Cedge{\lmono}{}{}{5-1}

			\Edge{3-3}{\rmono}{fill=white}{b_j^*}{3-5}
			\Cedge{\lmono}{fill=white}{b_i^*}{5-3}
			
			\Edge{1-5}{}{fill=white}{e_{P_i}}{3-5}
			\Edge{5-1}{}{fill=white}{e_{P_j}}{5-3}
			\Edge{3-5}{\lmono}{fill=white}{e_{P_j^*}}{5-5}
			\Edge{5-3}{\rmono}{fill=white}{e_{P_j^*}}{5-5}
			;
		\end{tikzpicture}
		\caption{consturction of an embedding out of the embedding of a rule overlap}
		\label{fig:compoistion_3}
	\end{figure}
	 
	Because $co$ is embedded in $(t_1,t_2)$ there is an morphism pair $(e_{P^*_i} \colon P_i^* \to P_i^*P_j^*, e_{P^*_j} \colon P^*_j \to P_i^*P_j^*) \in \Epi$ and a morphism $m^* \colon P_i^*P_j^* \mono G \in \Mono$ so that $m_i^* = m^* \circ e_{P_i^*}$ and $m_j^* = m^* \circ e_{P_j^*}$ by \EMfactorization. Moreover, \cref{lem:included_pullback} implies that the square (5) (shown in the diagram below) is a pullback (shown in \cref{fig:composition4}). 
	We continue by constructing the pullbacks $(2)$, $(3)$ and $(4)$. 
 	The overall square (2) + (3) + (4) + (5)  is also a pullback (as composition of pullbacks). 
 	If we choose $m_1^i \coloneq m_i^* \circ e_{P_i}$ and $m_2^j \coloneq m_j^* \circ e_{P_j}$, the overall squares commutes and \cref{lem:included_pullback} implies that the outer square (consisting of $m_1^i \circ b_i$ and $m_2^j \circ b_j$) is a pullback, i.e., $ro$ is embedded in $(t_1,t_2)$ via $m_1^i$ and $m_2^j$. 
 	Analogously, we can also show that $ro'$ is embedded in $(t_1,t_2)$. \qed
	
	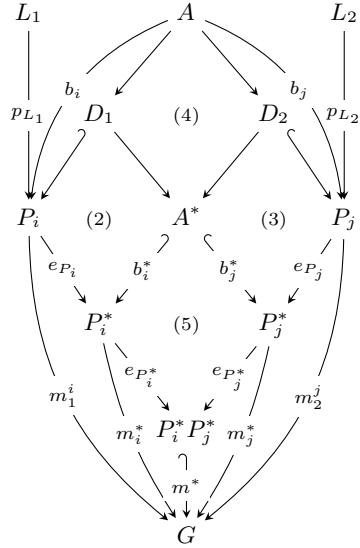
\begin{figure}
		\centering
		\begin{tikzpicture}
			\matrix (m) [	matrix of math nodes,
			nodes in empty cells,
			row sep=1em,
			column sep=1em,
			minimum width=1em]
			{
				L_1& & A_{} & & L_2 \\ \\
				& D_1 &\text{\scriptsize (4)} & D_2\\ \\
				P_i & \text{\scriptsize (2)} & A^*_{} & \text{\scriptsize (3)}& P_j \\ \\
				& P_i^* & \text{\scriptsize (5)} & P_j^* \\ \\ 
				& & P_i^*P_j^* \\ \\
				& & G \\
			};
			\path[-stealth] 
			
			\Edge{1-1}{}{fill=white}{p_{L_1}}{5-1}
			\Edge{1-5}{}{fill=white}{p_{L_2}}{5-5}

			\Edge{1-3}{}{}{}{3-2}
			\Cedge{bend right = 30}{fill=white}{b_i}{5-1}
			\Cedge{bend left = 30}{fill=white}{b_j}{5-5}
			\Cedge{}{}{}{3-4}
			
			\Edge{3-2}{\lmono}{}{}{5-1}
			\Cedge{}{}{}{5-3}
			
			\Edge{3-4}{\rmono}{}{}{5-5}
			\Cedge{}{}{}{5-3}
			
			\Edge{5-1}{}{fill=white}{e_{P_i}}{7-2}
			\Cedge{bend right = 25}{fill=white}{m_1^i}{11-3}
			
			\Edge{7-2}{}{fill=white}{e_{P^*_i}}{9-3}
			\Cedge{bend right = 10}{fill=white}{m_i^*}{11-3}

			\Edge{7-4}{}{fill=white}{e_{P^*_j}}{9-3}
			\Cedge{bend left = 10}{fill=white}{m_j^*}{11-3}
			
			\Edge{5-5}{}{fill=white}{e_{P_j}}{7-4}
			\Cedge{bend left = 25}{fill=white}{m_2^j}{11-3}
			
			\Edge{5-3}{\lmono}{fill=white}{b_i^*}{7-2}
			\Cedge{\rmono}{fill=white}{b_j^*}{7-4}
			
			\Edge{9-3}{\lmono}{fill=white}{m^*}{11-3}
			
			;
		\end{tikzpicture}
		\caption{composition of an embedding out of an embedding of a composed essence}\label{fig:composition4}
	\end{figure}
	\end{proof}

\begin{proof}[of \cref{thm:correctness_ce_conditions}]
	Given a pair of transformations $(t_1,t_2) = (\trans{G}{\rho_1}{m_1}{H_1}, \trans{G}{\rho_2}{m_2}{H_2})$ via rules $\rho_1 = \rleAC{1}$, and $\rho_2 = \rleAC{2}$. 
	To show the claim, we need to  show that (i) a symbolic conflict essence of $\rho_1$ and $\rho_2$ is embeddable in $(t_1,t_2)$ if the transformation pair is parallel dependent and that (ii) the embeddability of a symbolic conflict essence in $(t_1,t_2)$ implies parallel dependence.
	
	\begin{enumerate}
		\item We start by assuming that $(t_1,t_2)$ is parallel dependent and show that a symbolic conflict essence is embedded in $(t_1,t_2)$: 
		By completeness of conflict essences (\cref{cor:embedding_conflict_essences}), there is a symbolic conflict essence $(ce, ac_{ce}) \in \CEssAC{\rho_1}{\rho_2}$, so that $ce$ is embedded in $(t_1,t_2)$. 
		Therefore, it remains to show that the embedding morphisms $m_1^j$ and $m_2^i$ satisfy $ac_{ce}$. 
		If $ce$ is not ac-conflicting, $ac_{ce} = \true$ and $(m_1^j, m_2^i) \models ac_{ce}$.
		
		Otherwise, if $ce$ is ac-conflicting, \cref{thm:disjointness_ac-disregarding} implies that $(t_1,t_2)$ is ac-dis\-re\-gar\-ding parallel independent. 
		Moreover, there is an morphism $m^{**} \colon D \to G$ where $(1)$ is the figure below is a pushout (which exists, because either $b_j$ or $b_i$ is an $\Mono$-morphism) so that 
		$m_1^j = m^{**} \circ b_i'$ and $m_2^i = m^{**} \circ b_j'$.
		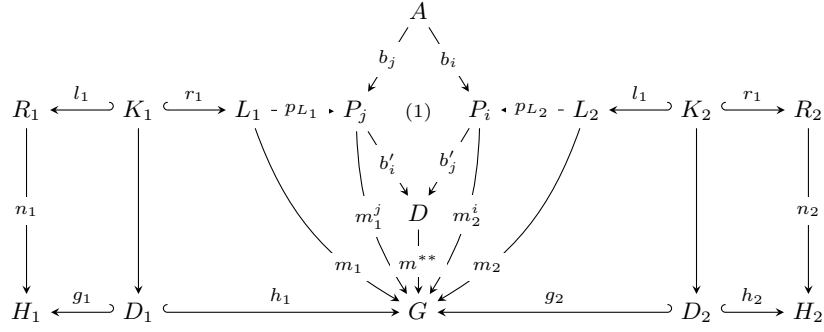
\begin{figure}[h]
			\centering
			\begin{tikzpicture}
				\matrix (m) [	matrix of math nodes,
				nodes in empty cells,
				row sep=1em,
				column sep=.8em,
				minimum width=1em]
				{	
					& & & &   & & &  A & & &  \\				\\			
					R_1 & & K_1 & & L_1 & & P_j & \text{\scriptsize(1)} & P_i & & L_2 & & K_2 & & R_2 \\ \\
					& & & & & & &  D \\				\\
					H_1 & & D_1 & &     & &  & G_{}&    & &     & & D_2 & & H_2 \\
				};
				\path[-stealth]

				(m-3-3) edge[\lmono] node [above] {\edge{l_1}}(m-3-1)
				edge[\rmono] node [above] {\edge{r_1}}(m-3-5)
				
				(m-7-3) edge[\lmono] node [above] {\edge{g_1}}(m-7-1)
				edge[\rmono] node [above] {\edge{h_1}}(m-7-8)
				
				(m-3-13) edge[\lmono] node [above] {\edge{l_1}}(m-3-11)
				edge[\rmono] node [above] {\edge{r_1}}(m-3-15)
				
				(m-7-13) edge[\rmono] node [above] {\edge{h_2}}(m-7-15)
				edge[\lmono] node [above] {\edge{g_2}}(m-7-8)
				
				(m-3-1) edge node [fill=white] {\edge{n_1}}(m-7-1)
				(m-3-3) edge node [] {}(m-7-3)
				
				(m-3-15) edge node [fill=white] {\edge{n_2}}(m-7-15)
				(m-3-13) edge node [] {}(m-7-13)
				
				(m-3-5) edge node [fill=white]{\edge{p_{L_1}}}(m-3-7)
				(m-3-5) edge[bend right = 20] node [fill=white, near end]{\edge{m_1}}(m-7-8)
				
				(m-3-11) edge node [fill=white]{\edge{p_{L_2}}}(m-3-9)
				(m-3-11) edge[bend left = 20] node [fill=white, near end]{\edge{m_2}}(m-7-8)
				
				(m-3-7) edge node [fill=white]{\edge{b_i'}} (m-5-8)
				edge[bend right = 15] node [fill=white]{\edge{m_1^{j}}} (m-7-8)
				
				(m-3-9) edge node [fill=white]{\edge{b_j'}} (m-5-8)
				edge[bend left = 15] node [fill=white]{\edge{m_2^{i}}} (m-7-8)
				
				(m-1-8) edge node [fill=white]{\edge{b_j}}(m-3-7)
				edge node [fill=white]{\edge{b_i}}(m-3-9)
				
				(m-5-8) edge[] node [fill=white]{\edge{m^{**}}}(m-7-8)
				;
			\end{tikzpicture}
			\caption{Embedding of a conflict essence in a transformation pair}\label{fig:embedding_conflict_essence}
		\end{figure}
		Because $(t_1,t_2)$ are parallel independent ac-diregardingly there is the extension shown in \cref{fig:ac_disregarding_initial_transformation} diagram where $m^* \colon L_1 + L_2 \to D$ is the coproduct morphism of $(L_1+ L_2, i_{L_1}, i_{L_2})$ so that 
		$m^* \circ i_{L_1} = b_i' \circ p_{L_1}$ and $m^* \circ i_{L_2} = b_j' \circ p_{L_2}$ \cite[Lemma 4]{LO21} (shown in \cref{fig:construction_conflict_essence}).

			\begin{figure}[h]
			\centering
			\begin{tikzpicture}
				\matrix (m) [	matrix of math nodes,
				nodes in empty cells,
				row sep=1em,
				column sep=1em,
				minimum width=1em]
				{
					R_1+L_2 & & L_1+L_2 & & L_1+ R_2 \\
					& & & & \\
					H_1 & & G_{} & & H_2 \\
				};
				\path[-stealth]
				(m-1-1) edge[] node [left]{}(m-3-1)
				(m-1-3) edge[-implies, double equal sign distance] node [above]{\edge{\rho_1,i_{L_1}}}(m-1-1)
				edge[-implies, double equal sign distance] node [above]{\edge{\rho_2,i_{L_2}}}(m-1-5)
				
				(m-1-3) edge[] node [fill=white]{\edge{m^{**} \circ m^{*}}}(m-3-3)
				
				(m-3-3) edge[-implies, double equal sign distance] node [above]{\edge{\rho_1, m_1}}(m-3-1)
				(m-3-3) edge[-implies, double equal sign distance] node [above]{\edge{\rho_2, m_2}}(m-3-5)
				
				(m-1-5) edge[] node []{}(m-3-5)
				;
				
			\end{tikzpicture}
			\caption{Extension diagram for ac-diregarding parallel independent transformations}\label{fig:ac_disregarding_initial_transformation}
		\end{figure}
		
		Because $(\trans{L_1 + L_2}{\rho_1}{i_{L_1}}{R_1+L_2}, \trans{L_1 + L_2}{\rho_2}{i_{L_2}}{R_1+L_2}, ac_{L_1 + L_2} \wedge ac^*_{L_1+L_2})$ is a critical pair for $\rho_1$ and $\rho_2$ \cite[Theorem 5]{LO21}, and the transformations are parallel dependent, $m^{**} \circ m^* \models ac_{L_1+ L_2} \wedge ac^*_{L_1 + L_2}$ \cite[Lemma 6.2]{EGHLO12}. 
		With correctness of the shift operator, we also have $m^{**} \models \Shift(m^*, ac_{L_1+L_2} \wedge ac^*_{L_1 + L_2})$.
		It follows that $(m_1^j, m_2^i) \models \exists(\Span{L_1}{b_i'}{D}{b_j'}{L_2}, \Shift(m^*, ac_{L_1+L_2})$, because we also have $m_1^j = m^* \circ b_i'$ and $m_2^i = m^* \circ b_j'$. 
		This means, that $(ce, ac_{ce})$ is embedded in $(t_1,t_2)$ via $m_1^j$ and $m_2^i$.
		
		\item Let us now assume that a symbolic conflict essence $(ce, ac_{ce}) \in \CEss{\rho_1}{\rho_2}$ is embedded in $(t_1,t_2)$ via morphisms $m_1^j \colon P_j \to G$ and $m_2^i \colon P_i \to G$. 
		We want to show $t_1$ and $t_2$ are parallel dependent.
		For this, we need to distinguish whether $(ce, ac_{ce})$ is ac-conflicting or not. 
		
		If $(ce, ac_{ce})$ is not ac-conflicting, \cref{thm:disjointness_ac-disregarding} implies that $t_1$ and $t_2$ are parallel dependent.
		
		Secondly, let us assume that $ce$ is ac-conflicting. 
		This means that the transformation pair $(t_1,t_2)$ is parallel independent ac-disregarding (\cref{thm:disjointness_ac-disregarding}).
		Again, there is an extension diagram  as shown in \cref{fig:ac_disregarding_initial_transformation}~\cite[Lemma 4]{LO21}. 
		Moreover, because (1) in \cref{fig:embedding_conflict_essence} is a pushout (which exists because $ce$ is ac-conflicting), the morphism $m^{**} \colon D \to G$, so that $m_1^j = m^{**} \circ b_i'$ and $m_2^i = m^{**} \circ b_j'$, is unique and because $(m_1^j, m_2^i) \models ac_{ce}$, we have $m^{**} \models \Shift(m^*, ac_{L_1+L_2} \wedge ac^*_{L_1+L_2})$. 
		This implies that $m^{**} \circ m^* \models ac_{L_1+L_2} \wedge ac^*_{L_1+L_2}$, i.e., the transformation pair $(t_1,t_2)$ is parallel dependent	\cite[Lemma 6.2]{EGHLO12}. \qed
		
	\end{enumerate}
\end{proof}

\begin{proof}[of \cref{relation_to_initial_conflicts}]
	To relate each symbolic conflict essence $sce = (ce, ac_{ce})$ to a unique initial conflict, we need to distinguish whether the included conflict essence is ac-conflicting or not. 
	Let us start by assuming that $ce$ is not ac-conflicting. 
	Here, $ce$ is either also a disabling essence or an composed essence. 
	If $ce$ is a disabling essence, the filtering in \cref{def:asynchronous_essences} implies that $ce$ is also a disabling for the plain rules and $ce$ can be uniquely related to an initial conflict~\cite{ACR19}. 
	If $ce$ is a composed essence, we know that by \cref{thm:embedding_conflict_essences}, both conflict essences $ce$ has been composed by are embedded in a transformation pair if $ce$ is embedded in it. 
	This means that both of these essences related to the same initial conflict, since at most one initial conflict embeds in a parallel dependent transformation pair~\cite{LO21}.
	Therefore, we can also relate $sce$ to this initial conflict.
	
	If $ce$ is ac-conflicting, an embedding of $ce$ in a transformation pair implies ac-disregarding parallel independence. 
	Now, the symbolic intial conflict is also embedded in this transformation pair~\cite[Lemma 4]{LO21} and because this the only initial conflict that embeds in ac-disregarding parallel independent transformation pairs, we can uniquely relate $sce$ to the symbolic initial conflict.
	
	The second part of the claim follows from the facts that an initial conflict can be embedded in each parallel dependent transformation pair and \cref{thm:correctness_ce_conditions} which shows the same statement for symbolic conflict essences. \qed
\end{proof}

	\section{Additional Examples}
\label{app:detailed_examples}

The following example is a more detailed version of \cref{ex:construction_deleting_essence}.
\begin{example}[Construction of a proto-essence by deletion]\label{ex:construction_deleting_essence_complete}
	\Cref{fig:ex_construction_proto_essence} shows the construction of a proto-essence by deletion of the rules \decapAttribute\ and \pullAttribute. 
	In this example, we chose the path $L_2, p_1, P_1, p_3,P_3$ as already described in \cref{ex:treeStructure}. 
	Please note that the graphs $P_1, P_2$ and $P_3$ now also show the included occurrence of the \ac{LHS} of \pullAttribute. 
	The construction starts with the overlap $L_1P_3$ of the \ac{LHS} $L_1$ of \decapAttribute\ and $P_3$. 
	Obviously, $L_1P_3$ is equal to $P_3$; the embedding of $L_1$ is implicitly described by the node identifiers. 
	For instance, the node \textsf{1:Class} is mapped to the node \textsf{1,10:Class} and so on. 
	The edges are mapped according to the mapping of their source and target nodes. 
	
	In the first step of the construction, we construct the pullback of $e_{L_1}$ and $e_{P_3} \circ p_3 \circ p_1$ which results in the graph $A_0$. 
	For graphs, the pullback object is the smallest subgraph of $L_1P_3$ so that elements of $L_1$ and $L_2$ overlap. 
	Clearly, for $L_1P_3$ no elements of the both \acp{LHS} overlap, which means that  $A_0$ is the empty graph. 
	We continue by calculating the pullback of $l_1$ and $A_0 \to L_1$ (which was obtained by the first pullback). 
	As $A_0$ is the empty graph, the resulting object $A_0'$ is also the empty graph. 
	Clearly, the morphism from $A_0'$ to $A_0$ is an isomorphism, i.e., we proceed by considering the interactions of \decapAttribute\ and $P_1$. 
	
	The results of the pullbacks in the second iteration are the graphs $A_1$ and $A_1'$. 
	$L_1$ and $P_1$ only overlap at the node \textsf{1,10:Class} in $L_1P_3$, and since \decapAttribute\ does not delete any nodes, the result $A'_1$ of the second pullback also contains the node \textsf{1,10:Class}. 
	Once again, there is an isomorphism from $A_1'$ to $A_1$ (the identity morphism) and we proceed with the next iteration.
	
	In the third iteration, we obtain the graphs $A_2$ and $A_2'$. 
	When the rule \decapAttribute\ is applied and overlaps with $P_3$ as described in the overlap $L_1P_3$, it will delete the nodes \textsf{3:Method}, \textsf{2:Method} and their incoming edges. 
	This means that the nodes \textsf{3,12:Method}, and \textsf{2,13:Method} contained in $A_2$ will be deleted. 
	In consequence, the result $A_2'$ of the second pullback will now contain the nodes \textsf{1,10:Class}, \textsf{4,11:Attribute} and the \textsf{variables} edge that connects them. 
	This, in particular, means that there is no isomorphism from $A_2'$ to $A_2$ and we have found a situation in which an application of \decapAttribute\ can delete an occurrence of $P_3$.

	\begin{figure}
		\includegraphics[width = \textwidth]{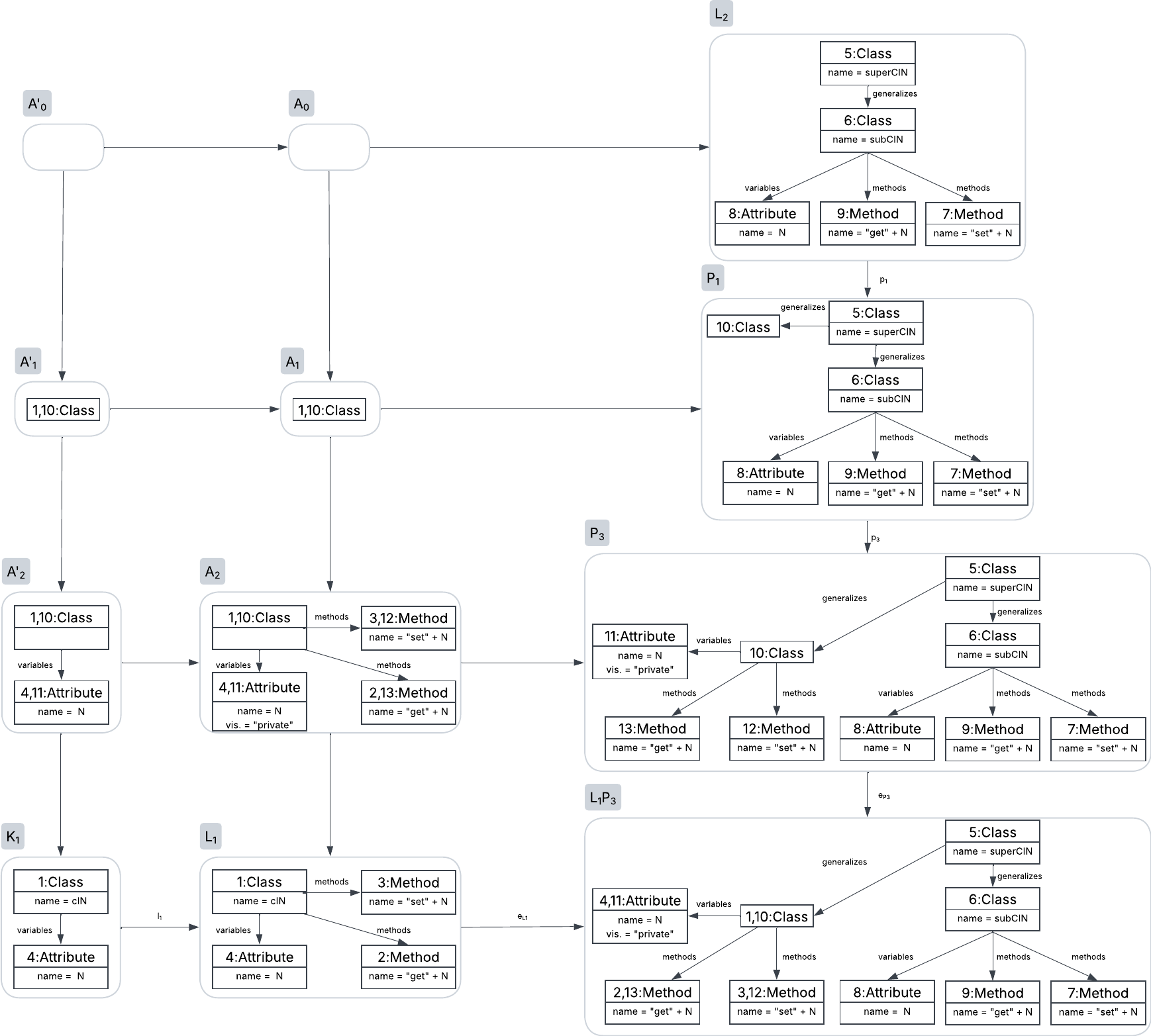}
		\caption{Construction of proto-essences by deletion of the rules \decapAttribute\ and \pullAttribute }\label{fig:ex_construction_proto_essence}
	\end{figure}

	In the final step of the computation, we calculate the initial pushout of the morphism $A_2' \to A_2$ as shown in \cref{fig:ex_construction_proto_essence_ipo}.
	Intuitively, the graphs $B$ and $C$ of the initial pushout describe the elements that are deleted by \decapAttribute\ and their \emph{boundary}, i.e., each node that is attached to a deleted edge, i.e., the \textsf{Method} nodes, their incoming edges, and the \textsf{Class} node. 
	Moreover, the graph $B$ shows the elements that are not deleted but needed in $C$, i.e., all those nodes that are not  deleted but have an incoming or outgoing edge that is deleted. 
	In this particular case, this only applies for the node \textsf{1,10:Class}.
	\begin{figure}
		\centering
		\includegraphics[width = 0.8\textwidth]{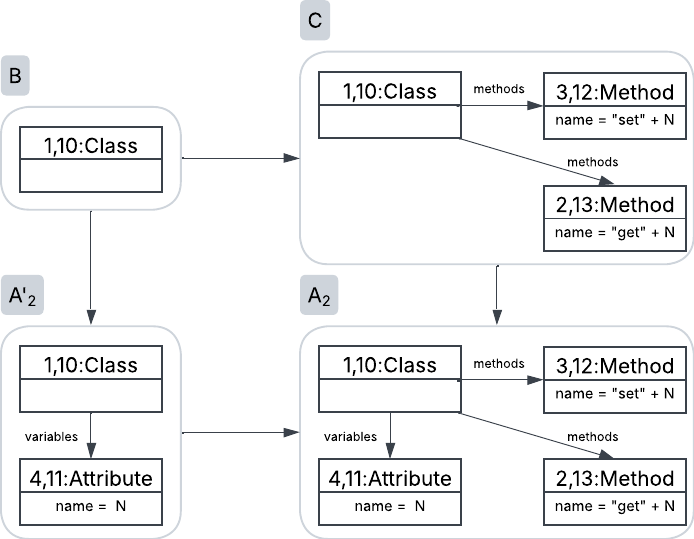}
		\caption{Initial pushout that concludes the construction of proto-essences by deletion of of the rules \decapAttribute\ and \pullAttribute}\label{fig:ex_construction_proto_essence_ipo}
	\end{figure}
	The computation of the initial pushout concludes the construction  and the resulting proto-essence $(a \circ c \colon C \to L_1, b \circ c \colon C \to P_3)$ is shown in \cref{fig:ex_proto_essence}.
	\begin{figure}
		\includegraphics[width= \textwidth]{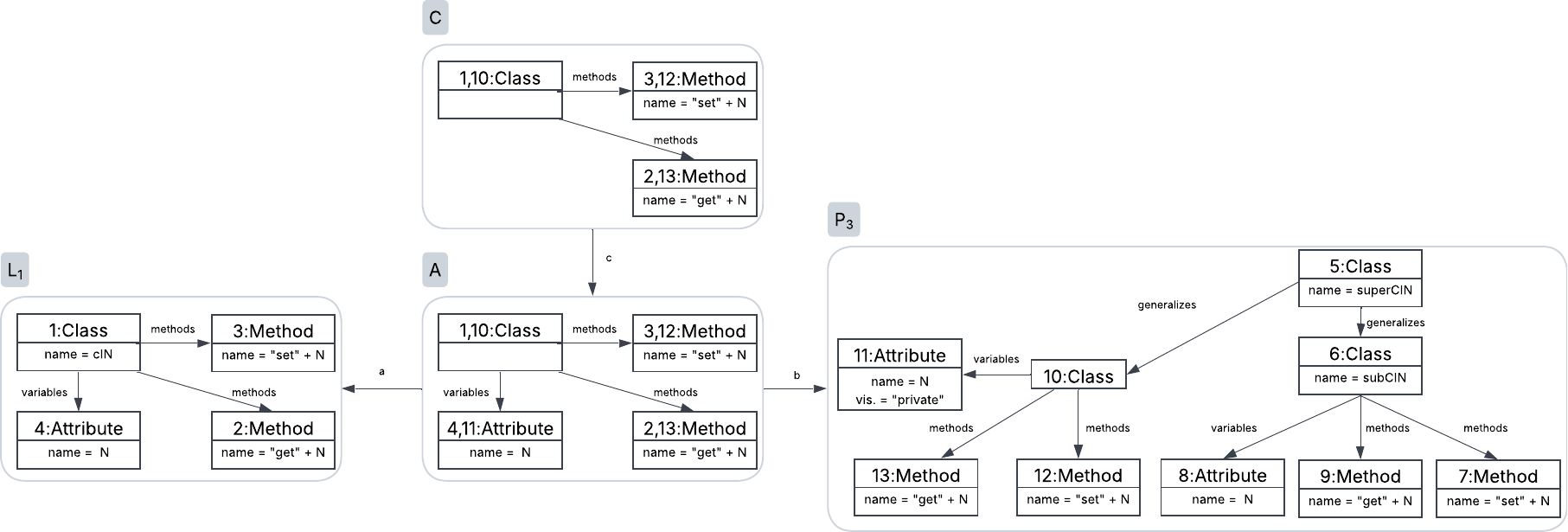}
		\caption{Proto-essence by deletion of the rules \decapAttribute\ and \pullAttribute}\label{fig:ex_proto_essence}
	\end{figure}
\end{example}

Applying the just illustrated computation to all possible overlaps, considering both directions and also insertions, provides the set of proto-essences of the rule pair \decapAttribute{} and \pullAttribute{}; in our example, all proto-essences turn out to be disabling essences. 
We discussed this set in \cref{ex:disabling_essences} at the start of \cref{chapter:essences}. 
The following example explains the individual essences in more detail. 
\begin{example}[Disabling essence]\label{ex:disabling_essences_complete}
	\Cref{fig:deleting_essences} shows six disabling essences of the  \emph{decapsulateAttribute} and \emph{pullUpEncapsulatedAttribute} rules. Moreover, they also represent rule overlaps of the rules.
	Only the objects $A$ (in the grey boxes) and $C$ (in the orange boxes) are displayed. 
	The morphisms to the rules' \acp{LHS} and the objects in the application condition are given through the node identifiers. 
	The graph contained in \textsf{$A^1_{L_1P_3}$} shows 
	that \emph{decapsulateAttribute} could violate the application condition of \emph{pullUpEncapsulatedAttribute} by destroying an occurrence of $P_3$, which is existentially bound in the application condition of \emph{pullUpEncapsulatedAttribute}. 
	In this disabling essence, the node \textsf{4,11:Attribute} is mapped to \textsf{4:Attribute} in the \ac{LHS} of \emph{decapsulateAttribute} and to \textsf{11:Attribute} in $P_3$.
	Whereas \textsf{$A^1_{L1P3}$} shows the elements that must be overlapped when this essence is embedded in a transformation pair, \textsf{$C^1_{L1P3}$} indicates the exact elements that cause a conflict: 
	When embedded in a transformation pair, \emph{decapsulateAttribute} destroys both \textsf{methods} edges leading from \textsf{1,10:Class} to \textsf{3,12:Method} and from \textsf{1,10:Class} to \textsf{2,13:Method}.
	The existence of an embedding of this disabling essence implies that an occurrence of $P_3$ is destroyed. 
	However, this does not imply that the application condition of \emph{pullUpEncapsulatedAttribute} is violated: 
	The matched class \textsf{1,10:Class} might contain other methods that, together with \textsf{1,10:Class}, form another occurrence of $P_3$.
	\textsf{$A^2_{L1P3}$} and \textsf{$A^3_{L1P3}$} show two other ways so that \emph{decapsulateAttribute} can destroy an occurrence of $P_3$. 
	Either \emph{decapsulateAttribute} only destroys the \textsf{method} edge to the getter-method of $P_3$ (shown in $A^3_{L1P3}$) or it only destroys the \textsf{method} edge to the setter-method of $P_3$ (shown in $A^2_{L1P3}$).

	In the middle row of \cref{fig:deleting_essences}, we display disabling essences indicating that \emph{decapsulateAttribute} causes a conflict by destroying the \ac{LHS} of \emph{pullUpEncapsulatedAttribute}, and vice versa. 
	Note that these essences are also calculated by the approach for plain rules as presented in~\cite{ACR19}. 
	When $A_{L1L2}$ is embedded in a transformation pair, \emph{decapsulateAttribute} destroys the match of \emph{pullUpEncapsulatedAttribute} by destroying the setter method and its incoming \textsf{methods} edge, which are also matched by \emph{pullUpEncapsulatedAttribute}. 
	$A^1_{L2L1}$ describes an similar conflict, but in this case it is the getter method and its edge that are destroyed.
	These two disabling essences also describe conflicts caused by \emph{pullUpEncapsulatedAttribute}, because \emph{pullUpEncapsulatedAttribute} destroys the respective \textsf{methods} edge, which is also matched by \emph{decapsulateAttribute}.
	Finally, $A^2_{L2L1}$ describes a conflict caused by \emph{pullUpEncapsulatedAttribute} that arises when \emph{pullUpEncapsulatedAttribute} destroys the \textsf{methods} edge leading to \textsf{4,8:Attribute}, which is also matched by \emph{decapsulateAttribute}.
	(The bottom row of \cref{fig:deleting_essences} will be explained in \cref{ex:composedRuleOverlap}.)
\end{example}

The following two examples illustrate constructions from \cref{chapter:embedding}. 
First, we illustrate the use of our criterion for checking rule overlaps to be ac-conflicting. 

\begin{example}[Ac-conflicting disabling essence]\label{ex:ac-conflicting-essence}
	In \cref{ex:ac-conflicting}, we check whether a conflict essence of \decapAttribute\ and \pullAttribute\ (shown in \cref{fig:deleting_essences}) is ac-conflicting. 
	For this, we compute the two pullbacks shown in \cref{ex:ac-conflicting}, where the first computes the elements of $P_3$ in which $A^3_{L_1P_3}$ and $L_2$ overlap. 
	The result is the graph $L_2'$ which is equal to the empty graph. 
	In the second step, we compute the pullback of $l_2$ and $L_2' \to L_2$ to check which elements will be deleted by \pullAttribute. 
	Since $L_2'$ is the empty graph, $K_2'$ is also equal to the empty graph, i.e., the disabling essence is ac-conflicting.	
	\begin{figure}
		\includegraphics[width= \textwidth]{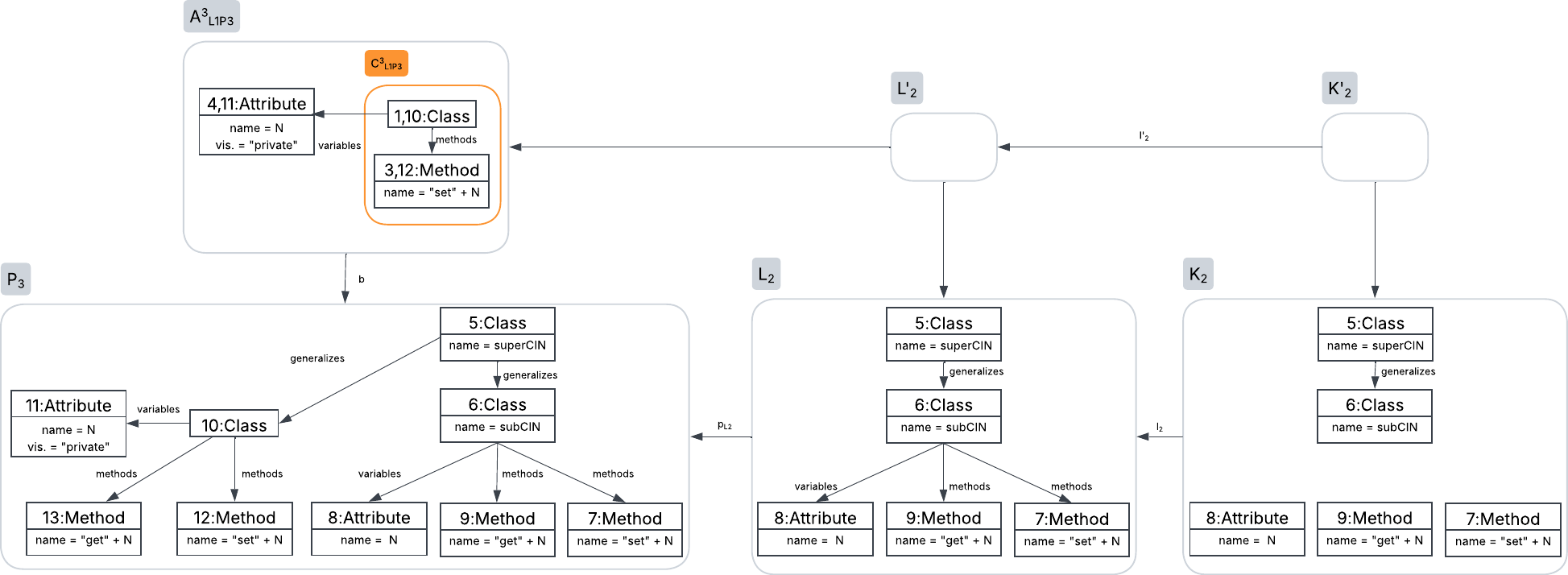}
		\caption{Check whether a disabling essence for the rules \decapAttribute\ and \pullAttribute\ is ac-conflicting}\label{ex:ac-conflicting}
	\end{figure}
\end{example}

The following extends \cref{ex:composedRuleOverlap} by providing a detailed computation of a composed rule overlap.
\begin{example}[Composed rule overlap]\label{ex:composedRuleOverlap1}
	\Cref{fig:deleting_essences} shows two composed rule overlaps of disabling essences of \decapAttribute\ and \pullAttribute. 
	\textsf{Comp($A_{L1L2}, A^1_{L2L1}$)} indicates that $A_{L1L2}$ and $A^1_{L2L1}$ can be embedded in the same transformation by identifying \textsf{1,6:Class}. 
	An occurrence of this composed overlap indicates that \emph{decapsulateAttribute} causes a conflict with \emph{pullUpEncapsulatedAttribute} by destroying \textsf{3,7:Method}, which is matched by \emph{pullUpEncapsulatedAttribute}. 
	It also indicates that \emph{pullUpEncapsulatedAttribute} causes a conflict with \emph{decapsulateAttribute} by destroying the edge leading to \textsf{2,9: Method}.
	When \textsf{Comp($A_{L1L2}, A^2_{L2L1}$)} is embedded in a transformation pair, \emph{decapsulateAttribute} will destroy the match of \emph{pullUpEncapsulatedAttribute} by deleting the edge leading to \textsf{3,7:Method} and \emph{pullUpEncapsulatedAttribute} will destroy the match of \emph{decapsulateAttribute} by deleting the edge leading to \textsf{4,8:Attribute}.

	\begin{figure}
		\includegraphics[width = \textwidth]{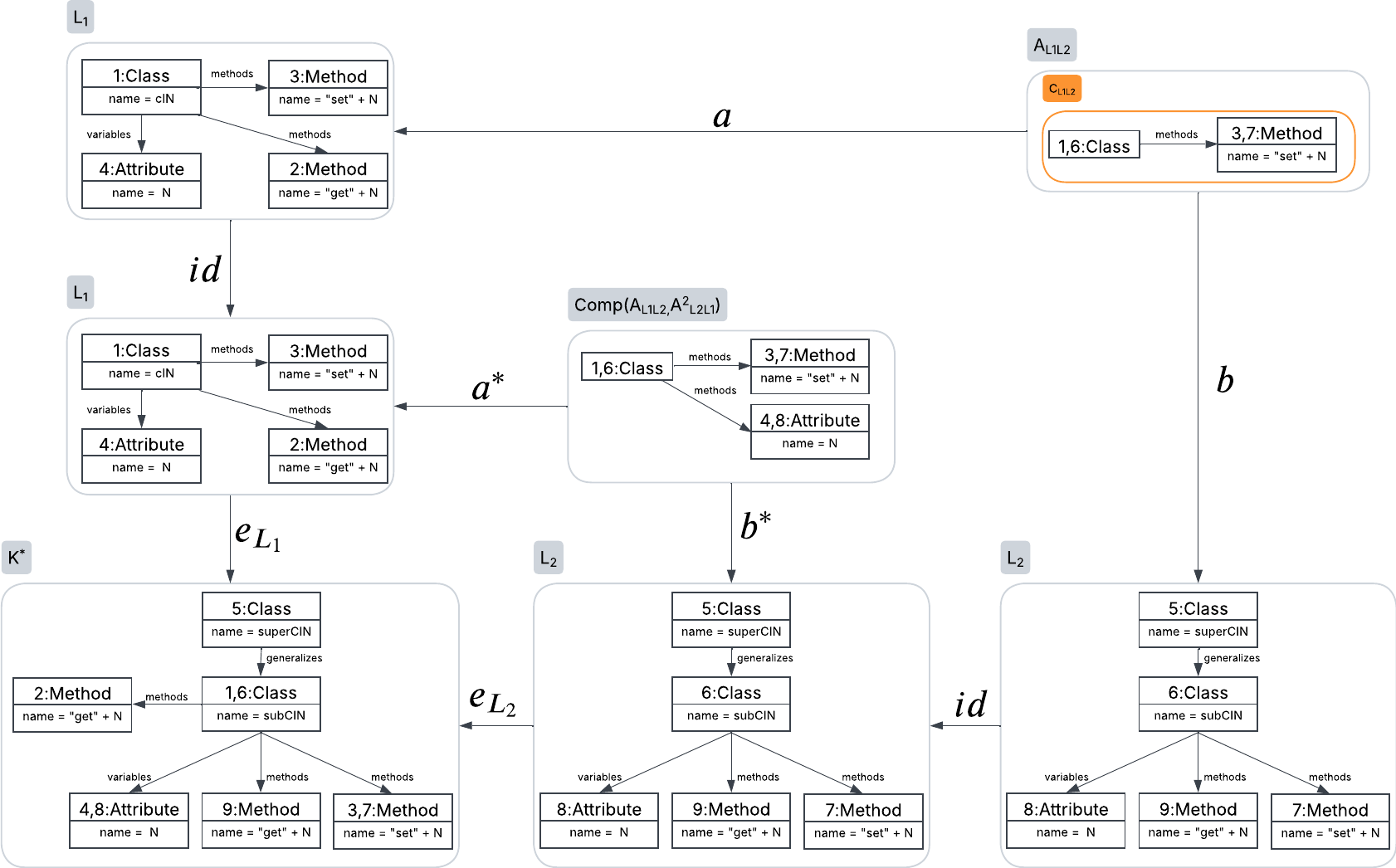}
		\caption{Composed disabling essence of the rule \decapAttribute\ and \pullAttribute}
		\label{fig:ex_composed_essences_construction}
	\end{figure}
	
	The exact computation of \textsf{Comp($A_{L1L2}, A^2_{L2L1}$)} is shown in \cref{fig:ex_composed_essences_construction}. 
	As both disabling essences involved signal a conflict for the plain rules, they are of the form 
	$$(a \colon C \to L_1, b \colon C \to L_2)$$ 
	This means that in the first step of the computation (when we are searching for an overlap of the codomains of $a$ and $a'$ and of $b$ and $b'$), we can simply chose $L_1$ and $L_2$, i.e., the \acp{LHS} of \decapAttribute\ and  \pullAttribute. 
	The second overlap $K^*$ is now an overlap of $L_1$ and $L_2$ as shown in the bottom left of \cref{fig:ex_composed_essences_construction}. 
	We now obtain the morphisms $a^* \in \Mono$ and $b^* \in \Mono$ by computing the pullback of $e_{L_1}$ and $e_{L_2}$.
	The result contains the three nodes \textsf{1,6:Class}, \textsf{3,7:Method}, \textsf{4,8:Class}, and their outgoing as these are the elements in $K^*$ that are matched by both $e_{L_1}(L_1)$ and $e_{L_2}(L_2)$.
\end{example}

\begin{example}[Symbolic conflict essence]\label{ex:detailed_symbolic_essence}
	Together with a cospan condition, each conflict essence shown in \cref{fig:deleting_essences} forms a symbolic conflict essence.
	For \textsf{A\textsubscript{L1L2}, A\textsuperscript{1}\textsubscript{L2L1}}, \textsf{A\textsuperscript{2}\textsubscript{L2L1}}, and the composed essences, 
	which are not ac-conflicting, the application condition is equal to $\true$. 
	The remaining conflict essences 
	are equipped with a non-trivial cospan condition of the form $$ac_{ce} \coloneq \exists(\CoSpan{P_i}{a'}{D}{b'}{P_j}, \Shift(m^*, ac_{L_1+L_2} \wedge ac^*_{L_1+L_2})),$$ where $d \coloneq \Shift(m^*, ac_{L_1+L_2} \wedge ac^*_{L_1+L_2})$ ensures that 
	\begin{enumerate}
		\item both rules satisfy their application conditions before the transformations (checked by $ac_{L_1 + L_2}$) and that
		\item the transformation pair is parallel dependent (checked by $ac^*_{L_1+L_2}$). This means that $ac^*_{L_1+L_2}$ checks whether both application conditions are satisfied after the transformations, or not. 
	\end{enumerate}
	Through the cospan, we can use both embedding morphisms of the symbolic conflict essence to check the satisfaction of $d$.
	
	\cref{fig:ex_symbolic_conflict_essence} shows a symbolic conflict essence for the rules \decapAttribute\ and \pullAttribute. 
	The contained conflict essence is displayed at the top, and the corresponding cospan condition at the bottom of the figure. 
	As explained in \cref{ex:disabling_essences}, the conflict essence signals that \decapAttribute\ deletes the getter and setter methods of the private attribute \textsf{4,11:Attribute} which could lead to a violation of the application condition of \pullAttribute\ if the class \textsf{1,10:Class} does not contain other getter or setter methods for this attribute, or a public attribute with the name \textsf{N}.
	
	The first row of the cospan condition contains the cospan that is constructed by calculating the pushout of the two morphisms contained in the conflict essence (compare \cref{def:symbolic_conflict_essences}). 
	In this case, the obtained graph $D$ is equal to $P_3$, and we just add a morphism from $L_1$ in $D$ accoring to the node identifiers. 
	Moreover, the graph $D$ is exactly the graph that has initially been used to construct the conflict essence $A^1_{L1P3}$. 
	The second row of the condition shows $ac_{L_1 + L_2}$ that checks whether both rules satisfy their respective application condition before the transformation. 
	As \decapAttribute\ has no application condition, we only need to check the application condition of \pullAttribute. 
	Since \pullAttribute\ will move the node \textsf{8:Attribute} and its getter and setter method from \textsf{6:Class} to \textsf{5:Class}, we need to check whether each subclass of \textsf{5:Class} contains an attribute with name \textsf{N} that is either public or has corresponding getter and setter methods contained in the same class. 
	This is why in $P_1'$ we are searching for any other subclass of \textsf{5:Class}. 
	Each of these classes must now contain an attribute with name \textsf{N} that is either 
	\begin{enumerate}
		\item public, as checked by $P_2'$ or 
		\item private and has getter and setter methods contained in the same method (as checked by $P_3'$). 
	\end{enumerate}
	Please note that we omit another graph that checks whether \textsf{1,10:Class} contains an public attribute with name \textsf{N} for space reasons. 
	This graph is essentially equal to $P_2'$ with the difference that, instead of in \textsf{16:Class}, the attribute \textsf{17:Attribute} is contained in \textsf{1,10:Class} and the class \textsf{16:Class} does not exist.
	
	The bottom row of the condition shows $ac^*_{L_1+L_2}$ that checks whether \pullAttribute\ will still satisfy its application condition after \decapAttribute\ was applied. 
	For this, we only need to check whether \textsf{1,10:Class} still satisfies the remaining part of the condition. 
	This is because \decapAttribute\ only changes structure contained in this class. 
	There are two possibilities so that \textsf{1,10:Class} still satisfies the condition, either 
	\begin{enumerate}
		\item it contains another attribute with name \textsf{N} that is public, as modelled in $P_4'$, or 
		\item it contains two other getter and setter methods for \textsf{4,11:Attribute} that are not deleted by \decapAttribute\ as modelled in $P_5'$.
	\end{enumerate}
	If this condition is satisfied, we know that \textsf{1,10:Class} still satisfies the condition after an application of \decapAttribute\ and, in particular, that \pullAttribute\ will satisfy its application condition after \decapAttribute\ was applied.
	
	\begin{sidewaysfigure}
		\includegraphics[width= \textwidth]{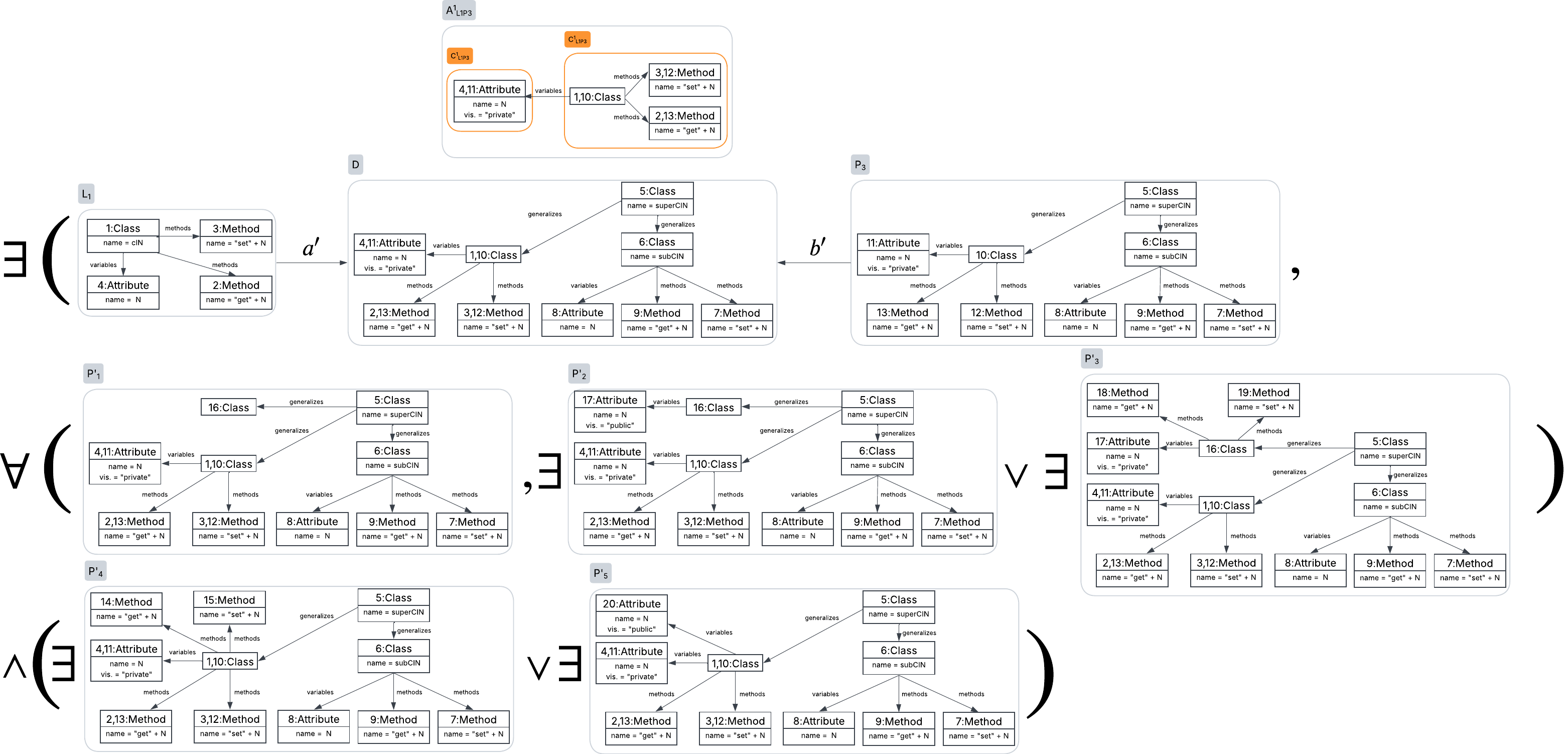}
		\caption{Symbolic conflict essence of the rules \decapAttribute\ and \pullAttribute}\label{fig:ex_symbolic_conflict_essence}
	\end{sidewaysfigure}
	
\end{example}

\end{document}